\documentclass[11pt,draftcls, onecolumn,journal]{IEEEtran}
\usepackage{amsmath,graphicx,times,color,amssymb,turnstile,extarrows}
\usepackage[noend]{algorithmic}
\usepackage{algorithm}
\usepackage{amsthm}
\usepackage{cite}
\def\Exp{{\mathbb E}}

\def\O{{\mathcal O}}

\def\bepsilon{{\boldsymbol \epsilon}}
\allowdisplaybreaks

\newtheorem{theorem}{Theorem}
\newtheorem{lemma}{Lemma}

\newtheorem{proposition}{Proposition}
\newtheorem{corollary}{Corollary}

\title{Performance Analysis and Scaling Law of MRC/MRT Relaying with CSI Error in
Massive MIMO Systems}
\author{Qian~Wang,~\IEEEmembership{Student Member,~IEEE,}
        and Yindi~Jing,~\IEEEmembership{Member,~IEEE}
\thanks{Qian Wang and Yindi Jing are with the Department
of Electrical and Computer Engineering, University of Alberta, Edmonton, Alberta, CA, e-mail:\{qw8,yindi\}@ualberta.ca}}%
\begin{document}

\maketitle
\vspace{-4mm}
\begin{abstract} 
This work provides a comprehensive scaling law and performance analysis for multi-user massive MIMO relay networks, where the relay is equipped with massive antennas and uses MRC/MRT for low-complexity processing. CSI error is considered. First, a sum-rate lower bound  is derived which manifests the effect of system parameters including the numbers of relay antennas and users, the CSI quality, and the transmit powers of the sources and the relay. Via a general scaling model on the system parameters with respect to the relay antenna number,  the asymptotic scaling law of the SINR as a function of the parameter scalings is obtained, which shows quantitatively the tradeoff between the network parameters and their effect on the network performance. 
In addition, a sufficient condition on the parameter scalings for the SINR to be asymptotically deterministic is given, which covers existing studies on such analysis as special cases. 
Then, the scenario where the SINR increases linearly with the relay antenna number is studied.  The sufficient and necessary condition on the parameter scaling for this scenario is proved. It is shown that in this case, the interference power is not asymptotically deterministic, then its distribution is derived, based on which the outage probability and average bit error rate of the relay network are analysed.
\end{abstract}


\section{Introduction}
\label{sec:intro}
 Multiple-input multiple-output (MIMO) systems refer to systems with multiple antennas implemented at the transceiver nodes. They exploit spatial diversity to provide high data rate and link reliability \cite{MIMO}. In conventional MIMO systems, the number of antennas is usually moderate (e.g., the LTE standard allows for up to 8 antenna ports). Recently, large-scale MIMO systems or massive MIMO, where hundreds of antennas are implemented at the transceiver nodes, attract a lot of attention \cite{massive,massive-1}. It has been shown that, due to the large scale, the antennas can form sharp beams toward desired terminals, thus providing high spectral and energy efficiencies. Besides, the effects of small-scale fading and interference can be significantly reduced with linear signal processing, such as maximal-ratio-combining (MRC), maximal-ratio-transmission (MRT), and zero-forcing (ZF) \cite{massive}. 

The performance of massive MIMO systems have been widely studied in the literature \cite{efficiency, MRCvsZF, SINR_distri, composite_fading, Chi, scaling-mimo-qi}. In \cite{efficiency}, for the uplink of massive MIMO systems with MRC or ZF, the deterministic equivalence of the achievable sum-rate is derived by using the law of large numbers. The following power scaling laws are shown. With perfect channel state information (CSI), the user and/or relay power can be scaled down linearly with the number of antennas while maintaining the same signal-to-interference-plus-noise-ratio (SINR); when there is CSI error (where minimum mean-squared error (MMSE) estimation is used) and the training power equals the data transmit power, the power can only be scaled down by the square root of the number of antennas. Another work on the energy efficiency and power efficiency of a single-cell multi-user massive MIMO network  is reported in \cite{MRCvsZF}, where a Bayesian approach is used to obtain the capacity lower bounds for both MRT and ZF precodings in the downlink. It is shown that that for high spectral efficiency and low energy efficiency, ZF outperforms MRT, while at low spectral efficiency and high energy efficiency the opposite holds. While the channel models used in \cite{efficiency,MRCvsZF} are Rayleigh fading, Ricean fading channel is considered in \cite{scaling-mimo-qi} in massive MIMO uplink, where the CSI is also obtained with MMSE estimator. Sum-rate approximations on the MRC and ZF receivers are obtained using the mean values of the components in the SINR formula. The derived power scaling law is that when the CSI is perfect or Ricean factor is non-zero, the user transmit power can be scaled down inversely proportional with the number of antennas while maintaining the same SINR level. Otherwise, the transmit power can only be scaled down inversely proportional  to the square root of the antenna number. 

While the aforementioned work analyses the sum-rate and power scaling law, there are also some work on the SINR distribution and outage probability. In \cite{SINR_distri}, the SINR probability density function (PDF) of MRT precoding is derived in closed-form in the downlink of a single-cell multi-user massive MIMO network. Besides, the asymptotic SINR performance is analysed when the number of users remains constant or scales linearly with the number of antennas. For the same network, in \cite{Chi}, the outage probability of MRT precoding is derived in closed-form. The authors first obtain the distribution of the interference power, based on which the outage probability is derived in closed-form. While only small-scale fading is considered in \cite{SINR_distri,Chi}, both small-scale (Rayleigh) fading and large-scale (log-normal) fading are considered in \cite{composite_fading}. In this work, the PDF of the SINR of MRC receiver is approximated by log-normal distribution, and the outage probability is derived in closed-form. The analysis shows that the shadowing effect cannot be eliminated by the use of a large number of antennas. 

Current results on massive MIMO show fantastic advantages of utilizing a large number of antennas in communications. A natural expansion of the single-hop massive MIMO systems is the two-hop massive MIMO relay networks, where the relay station is equipped with a large number of transmit and receive antennas to help the communications of multiple source-destination pairs. Relaying technology has been integrated to various wireless communication standards (e.g., LTE-Advanced and WIMAX Release 2) as it can improve the coverage and throughput of wireless communications \cite{relay-book}. Early studies focus on single-user relay networks and various relaying schemes, such as amplify-and-forward (AF) and decode-and-forward (DF), have been proposed \cite{relay-book}. With ever-increasing demands for higher performance, recently, multi-user relay networks have gained considerable attention \cite{M3,Qian1,Qian2,our-work1}. An important issue in multi-user relaying is how to deal with inter-user interference\cite{7of1}. By utilizing massive MIMO, the interference is expected to be significantly reduced and the network performance will be significantly improved.

Research activities on massive MIMO relay networks are increasing in recent years \cite{3,4,5,sum-rate-analysis,multi-cell, channel_aging,1,2,efficiency-zf,efficiency-twoway-mrcmrt-csi,efficiency-full-duplex-AF,rate-CSI}. In \cite{3,4}, for a single-user massive MIMO relay network with co-channel interferences at the relay, the ergodic capacity and outage probability of MRC/MRT and ZF relaying schemes are derived in closed-forms. The more general multiple-user massive MIMO relay networks are analysed in \cite{5,sum-rate-analysis,multi-cell, channel_aging,1,2,efficiency-zf,efficiency-twoway-mrcmrt-csi,efficiency-full-duplex-AF,rate-CSI}. Depending on the structure of the network model, the works can be divided to the following two categories. 

In \cite{sum-rate-analysis, multi-cell, channel_aging}, a network with multiple single-antenna users, one massive MIMO relay station and one massive MIMO destination is considered. This model applies to the relay-assisted uplink multiple-access network. In \cite{sum-rate-analysis}, it is shown that with perfect CSI, and infinite relay and destination antennas, the relay or user transmit power can scale inversely proportional to the number of antennas without affecting the performance. When there is CSI error, the user or relay power can only scale down with the square root of the number of antennas, given that the training power equals the transmit power. The same network is also considered in \cite{multi-cell, channel_aging} while the co-channel interference and pilot contamination are considered in \cite{multi-cell}, and channel aging effect is considered in \cite{channel_aging}. The effects of these factors on the power scaling are shown therein.

Another type of network is the relay-assisted multi-pair transmission network, where multiple single-antenna sources communicate with their own destinations with the help of a massive MIMO relay \cite{5,1,2,efficiency-zf,efficiency-twoway-mrcmrt-csi,efficiency-full-duplex-AF,rate-CSI}. In \cite{1,2}, the sum-rates of multi-pair massive MIMO relay network with MRC/MRT and ZF relaying under perfect CSI are analysed for one-way and two-way relaying respectively. In both work, with the deterministic equivalence analysis, it is shown that the sum-rate can remain constant when the transmit power of each source and/or relay scales inversely proportional to the number of relay antennas. In \cite{5}, the same network model as \cite{2} is considered for MRC/MRT relaying where the number of relay antennas is assumed to be large but finite. The analysis shows that, when the transmit powers of the relay and sources are much larger than the noise power, the achievable rate per source-destination pair is proportional to the logarithm of the number of relay antennas, and is also proportional to the logarithm of the reciprocal of the interferer number. In \cite{efficiency-full-duplex-AF}, the full-duplex model is considered for one-way MRC/MRT relaying and a sum-rate lower bound is derived with Jensen's inequality. 

While the above work assume perfect CSI at the relay, recent study has turned to networks with CSI error \cite{efficiency-twoway-mrcmrt-csi, efficiency-zf, rate-CSI}, which is more practical and challenging to analyse.
In \cite{efficiency-zf,rate-CSI}, a one-way massive MIMO relay network model is considered, where MMSE estimation is used to obtain the CSI. While  \cite{efficiency-zf} uses ZF relaying and assumes that the CSI error exists in both hops,  \cite{rate-CSI} uses MRC/MRT relaying and assumes that the CSI error only exists in the relay-destination hop. In both work, the power scalings of the sources and relay for non-vanishing SINR are discussed under the assumption that the training power equals the data transmission power. Compared with previous power scaling law results, the analysis in \cite{efficiency-zf,rate-CSI} are more comprehensive by allowing the power scaling to be anywhere between constant and linearly increasing with the number of relay antennas.  \cite{efficiency-twoway-mrcmrt-csi} is on a two-way MRC/MRT relaying network with CSI error.  With deterministic equivalence analysis, it is shown that when the source or relay power scales inversely proportional to the number of relay antennas, the effects of small-scale fading, self-interference, and noise caused by CSI error all diminish. 

In this work, the performance of MRC/MRT relaying in a one-way massive MIMO relay network with CSI error is investigated .  Our major differences from existing work are summarized as blow.

\begin{itemize}
\item Our system model is different from all the aforementioned existing work in relaying scheme, CSI assumption, or communication protocol. The work with the closest model is \cite{rate-CSI}, where the CSI error is assumed to exist in the relay-destinations hop only. We use a more general model where CSI error exists in both hops.

\item In our scaling law analysis, a general model for network parameters, including the number of source-destination pairs, the CSI quality parameter, the transmit powers of the source and the relay, is proposed. In this model, the scale exponent with respect to the relay antenna number can take continuous values from '0' to '1'. In most existing work, only a few discrete values for the power scaling, e.g., $0,1,1/2$, are allowed. Although \cite{efficiency-zf,rate-CSI} allow continuous exponent values, they constrains the number of sources as constant and the training power equals to the transmit power.

\item While in existing work, the asymptotically deterministic equivalence analysis is based on the law of large numbers, we use the quantized measure, squared coefficient of variation (SCV), to examine this property. As law of large numbers only applies to the summation of independent and identical distributed random variables, by using the SCV, we can discuss the asymptotically deterministic property of random variables with more complex structures. 

\end{itemize}

Based on these features that distinguish our work from existing ones, our unique contributions are listed as below. 
\begin{enumerate}
\item Firstly, by deriving a lower bound on the sum-rate, we investigate the performance scaling law with respect to the relay antenna number for a general setting on the scalings of the network parameters. The law 
provides comprehensive insights and reveals quantitatively the tradeoff among different system parameters.

\item Deterministic equivalence is an important framework for performance analysis of massive MIMO systems. We derive a sufficient condition on the parameter scales for the SINR to be asymptotically deterministic. Compared with existing work, where only specific asymptotic cases are discussed, our derived sufficient condition is more comprehensive. It covers all cases in existing works, and shows more asymptotically deterministic SINR scenarios. Besides, for the SINR to be asymptotically deterministic, the tradeoff between different parameter scales is also discussed.   

\item Through the scaling law results, we show that for practical network scenarios, the average SINR is at the maximum linearly increasing with the number of relay antennas. We prove that the sufficient and necessary condition for it is that all other network parameters remain constant. Furthermore, our work shows that in this case the interference power does not diminish and it dominates the statistical performance of the SINR. By deriving the PDF of the interference power in closed-form, expressions for outage probability and average bit error rate (ABER) are obtained.
While existing work mainly focus on the constant SINR case, this linearly increasing SINR case, suitable for high quality-of-service applications, has not been studied.
\end{enumerate} 

The remaining of the paper is organized as follows. In the next section, the system model including both the channel estimation and data transmission under MRC/MRT relaying is introduced. Then the performance scaling law is analyzed in Section \ref{sec:scaling}. In Section \ref{sec:deter}, the asymptotically deterministic SINR case is discussed. The linearly increasing SINR case is investigated in Section \ref{sec:linear}.  Section \ref{sec:simu} shows the simulation results and Section \ref{sec:con} contains the conclusion.

\section{System Model and Preliminaries for Scaling Law Analysis}
\label{sec:model}
We consider a multi-pair relay network with $K$ single-antenna sources ($S_1,\cdots,S_K$), each transmitting to its own destination. That is, $S_i$ sends information to Destination $i$, $D_i$.  
We assume that the sources are far away from the destinations so that no direct connections exist. To help the communications, a relay station is deployed \cite{relay-book}. The number of antennas at the relay station, $M$, is assumed to be large, e.g., a few hundreds \cite{3,4,5,sum-rate-analysis,multi-cell, channel_aging,1,2,efficiency-zf,efficiency-twoway-mrcmrt-csi,efficiency-full-duplex-AF,rate-CSI}. In addition, we assume $M\gg K$ because under this condition, simple linear relay processing, e.g., MRC/MRT, can have near optimal performance in massive MIMO systems \cite{non-cooperative}.

Denote the $M\times K$ and $K\times M$ channel matrices of the source-relay and relay-destination links as ${\bf F}$ and ${\bf G}$, respectively. The channels are assumed to be independent and identically distributed (i.i.d.) Rayleigh fading, i.e.,  entries of ${\bf F}$ and ${\bf G}$ are mutually independent following the circular symmetric complex Gaussian (CSCG) distribution with zero-mean and unit-variance, denoted as $\mathcal{CN}(0,1)$. The assumption that the channels are mutually independent is valid when the relay antennas are well separated. The information of ${\bf F}$ and ${\bf G}$ is called the channel state information (CSI), which is essential for the relay network. In practice, the CSI is obtained through channel training. Due to the existence of noises and interference, the channel estimation cannot be perfect but always contains error. The CSI error is an important issue for massive MIMO systems \cite{efficiency,scaling-mimo-qi,efficiency-twoway-mrcmrt-csi, efficiency-zf, rate-CSI}. In what follows, we will first describe the  channel estimation model, then the data transmission and MRC/MRT relaying scheme will be introduced.

\subsection{Channel Estimation}
To combine the received signals from the sources and precode the signals for the destinations, the relay must acquire CSI. ${\bf F}$, the uplink channel from the sources to the relay, can be estimated by letting the sources send pilots to the relay. In small-scale MIMO systems, ${\bf G}$ can be estimated by sending pilots from the relay to the destinations and the destinations will feedback the CSI to the relay \cite{MIMO,relay-book}. However, this strategy is not viable for massive MIMO systems, as the training time length grows linearly with the number of relay antennas $M$, which may exceed the channel coherence interval. Consequently, to estimate $\bf G$, we assume a time-division-duplexing (TDD) system with channel reciprocity \cite{massive}. So pilots are sent from the destinations and the relay-destination channels can  be estimated at the relay station.

 Without loss of generality, we elaborate the estimation of $\bf F$, and the estimation of $\bf G$ is similar.
Since the channel estimation is the same as that in the single-hop MIMO system, we will briefly review it and more details can be found in \cite{MIMO,relay-book} and references therein. Denote the length of the pilot sequences as $\tau$. For effective estimation, $\tau$ is no less than the number of sources $K$ \cite{efficiency,MRCvsZF}.  Assume that all nodes use the same transmit power for training, which is denoted as $P_t$. Therefore, the pilot sequences from all $K$ sources can be represented by a $\tau\times K$ matrix $\sqrt{\tau P_t}{\bf \Phi}$, which satisfies ${\bf \Phi}^H{\bf \Phi}={\bf I}_K$. The $M\times \tau$ received pilot matrix at the relay is 
\begin{equation*}
{\bf Y}_{train}=\sqrt{\tau P_t}{\bf F}{\bf \Phi}^T+{\bf N},
\end{equation*}
where ${\bf N}$ is the $M\times \tau$ noise matrix with i.i.d. $\mathcal{CN}(0,1)$ elements. 

The MMSE channel estimation is considered, which is widely used in the channel estimation of massive MIMO networks \cite{efficiency, sum-rate-analysis,efficiency-zf, scaling-mimo-qi}. The MMSE estimation of ${\bf F}$ given ${\bf Y}_{train}$ is 
\begin{equation*}
\hat{{\bf F}}=\frac{1}{\sqrt{\tau P_t}}{\bf Y}_{train}{\bf \Phi}^*\frac{\tau P_t}{1+\tau P_t}=\frac{\tau P_t}{1+\tau P_t}\left({\bf F}+\frac{1}{\sqrt{\tau P_t}}  {\bf N}_F\right),
\end{equation*}
where ${\bf N}_F\triangleq {\bf N}{\bf \Phi}^*$, which has i.i.d. $\mathcal{CN}(0,1)$ elements. Similarly, the MMSE estimation of ${\bf G}$ is 
\begin{equation*}
\hat{{\bf G}}=\frac{\tau P_t}{1+\tau P_t}\left({\bf G}+\frac{1}{\sqrt{\tau P_t}}  {\bf N}_G\right).
\end{equation*}
Define ${\bf E}_f\triangleq \hat{\bf F}-{\bf F}$ and ${\bf E}_g\triangleq \hat{\bf G}-{\bf G}$ which are the estimation error matrices. Due to the feature of MMSE estimation, $\hat{\bf F}$ and ${\bf E}_f$, $\hat{\bf G}$ and ${\bf E}_g$ are mutual independent. Elements of $\hat{\bf F}$ and $\hat{\bf G}$ are distributed as $\mathcal{CN}(0,\frac{\tau P_t}{\tau P_t +1})$. Elements of ${\bf E}_f$ and ${\bf E}_g$ are distributed as $\mathcal{CN}(0,\frac{1}{\tau P_t +1})$. 

Define
\begin{equation}
E_t\triangleq\tau P_t \text { and } P_c\triangleq \frac{\tau P_t}{\tau P_t+1}. \label{eq:pc}
\end{equation} 
So $E_t$ is total energy spent in training. $P_c$ is the power of the estimated channel element, representing the quality of the estimated CSI, while $1-P_c$ is the power of the CSI error. It is straightforward to see that $0\le P_c\le 1$. When $P_c\rightarrow 1$,  the channel estimation is nearly perfect. When $P_c\rightarrow 0$, the quality of the channel estimation is very poor. Note that, different combinations of $\tau$ and $P_t$ can result in the same $P_c$. For the majority of this paper, $P_c$ will be used in the performance analysis instead of $\tau$ and $P_t$. This allows us to isolate the training designs and focus on the effects of CSI error on the system performance. When we consider special cases with popular training settings, e.g., $\tau=K$ and the same training and data transmission power, $\tau$ and $P_t$ will be used instead of $P_c$ in modelling the CSI error.

\subsection{Data Transmissions}
With the estimated CSI, the next step is the data transmission. Various relay technologies have been proposed \cite{relay-book}. For massive MIMO systems, the MRC/MRT relaying is a popular one due to its computational simplicity, robustness, and high asymptotic performance \cite{1,2,3,4,5,efficiency-twoway-mrcmrt-csi, rate-CSI}. In the rest of this section, the data transmission with MRC/MRT relaying will be introduced.

Denote the data symbol of $S_i$ as $s_i$ and the vector of symbols from all sources as $\bf{s}$. With the normalization $\Exp(|s_i|^2)=1$, we have $\Exp({\bf s}^H{\bf s})=K$, where $(\cdot)^H$ represents the Hermitian of a matrix or a vector. Let $P$ be the average transmit power of each source. The received signal vector at the relay is
\begin{equation}
{\bf x}=\sqrt{P}{\bf F}{\bf s}+{\bf n_r},
\label{eq:x}
\end{equation}
where ${\bf n_r}$ is the noise vector at the relay with i.i.d.~entries each following $\mathcal{CN}(0,1)$.

With MRC/MRT relaying, 
the retransmitted signal vector from the relay is $a_e\hat{\bf G}^H\hat{\bf F}^H {\bf x}$, where $a_e$ is to normalize the average transmit power of the relay to be $Q$. 
With straightforward calculations, we have 
\begin{eqnarray}
a_e^2= \frac{Q}{\Exp\{{\rm tr}\left((\hat{{\bf G}}^H\hat{{\bf F}}^H {\bf x})(\hat{{\bf G}}^H\hat{{\bf F}}^H {\bf x})^H\right)\}}\approx\frac{Q}{PKP_c^3M^3(1+\frac{K}{MP_c}+\frac{1}{PP_cM})},
\label{eq:a_e}
\end{eqnarray} 
where the approximation is made by ignoring the lower order terms of $M$.

Denote ${\bf f}_i$, $\hat{{\bf f}}_i$, and ${\boldsymbol \epsilon}_{f,i}$ as the $i$th columns of ${\bf F}$, $\hat{\bf F}$ and ${\bf E}_f$ respectively; ${\bf g}_i$, $\hat{\bf g}_i$ and ${\bepsilon}_{g,i}$ as the $i$th rows of ${\bf G}$, $\hat{\bf G}$ and ${\bf E}_g$ respectively. The received signal at $D_i$ can be written as follows. 
\begin{eqnarray}
y_i&=&a_e\sqrt{P}{\bf g}_i\hat{{\bf G}}^H\hat{{\bf F}}^H{\bf F}{\bf s}+a_e{\bf g}_i\hat{{\bf G}}^H\hat{{\bf F}}^H{\bf n_r}+n_{d,i},\nonumber\\
&=&\underbrace{a_e\sqrt{P}\hat{\bf g}_i\hat{\bf G}^H\hat{\bf F}^H\hat{\bf f}_i s_i}_{\text{desired signal}} +\underbrace{a_e\sqrt{P}\sum_{k=1,k\neq i}^{K}{\bf g}_i\hat{\bf G}^H\hat{\bf F}^H{\bf f}_k s_k}_{\text{multi-user interference}} +\underbrace{a_e{\bf g}_i\hat{\bf G}^H\hat{\bf F}^H{\bf n_r}}_{\text{forwarded relay noise}}+\nonumber\\
&&\underbrace{a_e\sqrt{P}{\bf \epsilon}_{g,i}\hat{\bf G}^H\hat{\bf F}^H{\bepsilon}_{f,i} s_i-a_e\sqrt{P}\hat{\bf g}_{i}\hat{\bf G}^H\hat{\bf F}^H{\bepsilon}_{f,i} s_i-a_e\sqrt{P}{\bf \epsilon}_{g,i}\hat{\bf G}^H\hat{\bf F}^H\hat{\bf f}_{i} s_i}_{\text{noise due to CSI error}} +n_{d,i},\label{eq:recei_sig_e}
\end{eqnarray}
where $n_{d,i}$ is the noise at the $i$th destination following $\mathcal{CN}(0,1)$. Equation (\ref{eq:recei_sig_e}) shows that the received signal is composed of 5 parts: the desired signal, the multi-user interference, the forwarded relay noise, the CSI error term, and the noise at $D_i$.

Define 
\begin{eqnarray}
&&P_{s,e}\triangleq \frac{|\hat{\bf g}_i\hat{\bf G}^H\hat{\bf F}^H\hat{\bf f}_i|^2}{M^4}, \quad \quad \hspace{10mm}P_{i,e}\triangleq \frac{1}{K-1}\sum_{k=1,k\neq i}^{K}\frac{|{\bf g}_i\hat{\bf G}^H\hat{\bf F}^H{\bf f}_k|^2}{M^3},\label{eq:component1}\\ 
&&P_{n,e}\triangleq\frac{||{\bf g}_i\hat{\bf G}^H\hat{\bf F}^H||_F^2}{M^3}, \quad \quad \hspace{10mm}
P_{e,1}\triangleq \frac{(1-P_c)^2}{M^3}\sum_{n=1}^{K}\sum_{m=1}^{K}\hat{\bf f}_n^H\hat{\bf f}_m\hat{\bf g}_m\hat{\bf g}_n^H,\label{pse}\\
&&P_{e,2}\triangleq(1-P_c)\frac{\|\hat{\bf g}_{i}\hat{\bf G}^H\hat{\bf F}^H\|_F^2}{M^3}, \quad P_{e,3}\triangleq(1-P_c)\frac{\|\hat{\bf G}^H\hat{\bf F}^H\hat{\bf f}_{i}\|_F^2}{M^3}.\label{eq:component2}
\end{eqnarray}
From (\ref{eq:recei_sig_e}), we know that $P_{s,e},P_{i,e},P_{n,e}$ and $P_{e,1}+P_{e,2}+P_{e,3}$ are the normalized powers of the signal, the interference, the forwarded relay noise, and the noise due to CSI error respectively. With these definitions, the SINR of the $i$th source-destination pair can be written as
\begin{eqnarray}
{\rm SINR}_{i}= M\frac{P_{s,e} }{(K-1)P_{i,e}+\frac{1}{P} P_{n,e}+P_{e,1}+P_{e,2}+P_{e,3}+\frac{KP_c^3(1+\frac{K}{MP_c}+\frac{1}{PP_cM})}{Q}}.
\label{eq:SINR_e}
\end{eqnarray} 
The achievable rate for the $i$th source-destination pair is 
\begin{equation}
C_{i}=\Exp\left\{\frac{1}{2}\log_2(1+{\rm SINR}_{i})\right\}.
\end{equation}

\subsection{Preliminaries for Scaling Law Analysis}
This paper is on the performance behaviour and asymptotic performance scaling law of the massive MIMO relay network. It is assumed throughout the paper that the number of relay antennas $M$ is very large and the scaling law is obtained by studying the highest-order term with respect to $M$. 

Due to the complexity of the network, it is impossible to rigorously obtain insightful forms for the SINR and the achievable rate for the general $M$ case. Instead, we find the asymptotic  performance properties for very large $M$ with the help of Lindebergy-L$\acute{\text{e}}$vy central limit theorem (CLT). The CLT states that, for two length-$M$ independent column vectors ${\bf v}_1$ and ${\bf v}_2$, whose elements are i.i.d. zero-mean random variables with variances $\sigma_1^2$ and $\sigma_2^2$, $$\frac{1}{\sqrt{M}}{\bf v}_1^H{\bf v}_2\xrightarrow[]{d}\mathcal{CN}(0,\sigma_1^2\sigma_2^2),$$
where $\xrightarrow[]{d}$ means convergence in distribution when $M\rightarrow \infty$. 

Another important concept in the performance analysis of massive MIMO systems is \textit{asymptotically deterministic}. In many existing literature on massive MIMO, a random variable sequence $X_M$ is said to be asymptotically deterministic if it converges almost surely (a.s.) to a deterministic value $x$, i.e., 
\begin{equation*}
X_M \overset{a.s.}{\longrightarrow}x \text{ when } M\rightarrow \infty. 
\end{equation*} 
The strong law of large numbers is usually used to derive the deterministic equivalence. The almost sure convergence  implies the convergence in probability \cite{random_book}. Another type of convergence that implies convergence in probability is the convergence in mean square \cite{random_book}.  
For a random variable sequence $X_M$ with a bounded mean, $X_M$ converges in mean square to a deterministic value $x$, i.e., $X_M\overset{m.s.}{\longrightarrow}x$ if
\[
\lim_{M\rightarrow \infty}{\rm Var}\{X_M\}=0.
\]
 The convergence in mean square requires the variances of the random variable sequence to approach zero. It has been used to define the channel hardening effects for massive MIMO\cite{massive-1,hardening}, where the convergence in mean square means that the effects of small-scale fading is ignorable when the number of antennas is large. Besides, compared with almost sure convergence, the convergence in mean square is more tractable for analysis. We adopt the convergence in mean square for the asymptotically scaling law of massive MIMO relay network.   

However, the use of the variance may cause inconvenience and sometimes confusion in performance analysis of massive MIMO systems. One can always scale $X_M$ by $1/M^n$ with large enough $n$ to have the asymptotic deterministic property and the scaled random variable converges in mean square to 0. But this does not help the performance analysis when the scaling factor $M^n$ is put back into the SINR formula. Thus to avoid the scaling ambiguity, we use the squared coefficient of variance (SCV), defined as the square of the ratio of the standard deviation over the mean of the random variable \cite{scv}. 
It is noteworthy that the bounded mean condition is important. Without this condition, the convergence with $M\rightarrow\infty$ may not be well defined. Thus in this work, a random variable sequence $X_M$ with bounded mean is said to be asymptotically deterministic if 
\begin{equation}
\lim_{M\rightarrow \infty}{\rm SCV}\{X_M\}=0.
\end{equation}


\section{Analysis on the  Achievable Rate Scaling Law}
\label{sec:scaling}
The general performance scaling law of the massive MIMO relay network
 will be studied in this section.  
We start with analysing components of the received SINR to obtain a large-scale approximation. Consequently, a lower bound on the sum-rate is derived via Jensen's inequality. Then the performance scaling law and conditions for favourable SINR (non-decreasing SINR with respect to $M$) are derived. Typical network scenarios are discussed. Our analysis will show the relationship between the SINR scale and the parameter scales, and the trade-off between different parameter scales. 

\subsection{Sum-Rate Lower Bound and Asymptotically Equivalent SINR}
For the SINR analysis, we first derive the means and SCVs of components of the SINR, i.e., $P_{s,e}$, $P_{i,e}$, $P_{n,e}$, $P_{e,1}$, $P_{e,2}$ and $P_{e,3}$. 
With the help of CLT and tedious derivations, the following can be obtained.
\begin{eqnarray}
&&\Exp\{P_{s,e}\}\approx P_c^4, \quad \hspace{3.2cm}{\rm SCV}\{P_{s,e}\}\approx\frac{8}{M},\label{eq:scv1}\\
&&\Exp\{P_{i,e}\}\hspace{1mm}\approx \hspace{1mm}P_c^3\left(2+\frac{K}{MP_c}\right),\hspace{1.1cm}  
{\rm SCV}\{P_{i,e}\}\hspace{1mm}\approx \hspace{1mm}\frac{\frac{4}{K-1}+
\frac{8+10P_c}{P_cM}+\frac{K^2+18(K-2)P_c}{(K-1)P_c^2M^2}}{4+\frac{K^2}{M^2P_c^2}+\frac{4K}{MP_c}},\\
&&\Exp\{P_{n,e}\}\approx P_c^3+\frac{K}{M}P_c^2,\quad \hspace{1.6cm}
{\rm SCV}\{P_{n,e}\}\approx\frac{2+5P_c-2P_c^2}{MP_c+\frac{K^2}{MP_c}+2K},\\
&&\Exp\{P_{e,1}\}\approx\frac{K}{M}P_c^2(1-P_c)^2,\quad \hspace{1.1cm}{\rm SCV}\{P_{e,1}\}\approx\frac{3}{K},\\
&&\Exp\{P_{e,2}\}=\Exp\{P_{e,3}\}\approx P_c^3(1-P_c), \hspace{0.35cm} {\rm SCV}\{P_{e,2}\}={\rm SCV}\{P_{e,3}\}\approx 1,
\end{eqnarray}
where the approximations are made by keeping the dominant terms of $M$. Due to the space limit, we only show the derivations of $\Exp\{P_{s,e}\}$ and ${\rm SCV}\{P_{s,e}\}$ in Appendix A. The rest can be derived similarly. 

With our definitions in (\ref{eq:component1})-(\ref{eq:component2}) and by noticing that $P_c\in[0,1]$, the random variables $P_{s,e}$, $P_{i,e}$, $P_{n,e}$, $P_{e,1}$, $P_{e,2}$, $P_{e,3}$ all have bounded means. From (\ref{eq:scv1}), we know that $P_{s,e}$ is asymptotically deterministic since its SCV approaches to 0 as $M\rightarrow \infty$. Furthermore, the decreasing rate of its SCV is linear in $M$, showing a fast convergence rate. Thus, for large $M$, we can approximate it with its mean value. While for the rest components in the SINR, their SCVs  depend on the scalings of network parameters (such as $K$ and $P_c$), which do not necessarily converge to $0$. We cannot assume they are asymptotically deterministic so far. With the aforementioned approximation, the SINR expression becomes
\begin{eqnarray}
{\rm SINR}_{i}\approx \frac{MP_c^4 }{(K-1)P_{i,e}+\frac{1}{P} P_{n,e}+P_{e,1}+P_{e,2}+P_{e,3}+\frac{KP_c^3(1+\frac{K}{MP_c}+\frac{1}{PP_cM})}{Q}}.
\label{eq:SINR_e_app}
\end{eqnarray}
With this simplification,  the following result on the sum-rate can be obtained.
\begin{lemma}
The achievable rate of Source $i$ in the massive MIMO relay network has the following lower bound: 
\begin{eqnarray}
C_{i}\ge C_{i,LB} \triangleq \frac{1}{2}\log_2\left(1+\widetilde{\rm SINR}_{i}\right),\label{rate-LB}
\end{eqnarray}
where 
\begin{equation}
\widetilde{\rm SINR}_{i}\triangleq
\frac{1}{\frac{2K}{MP_c}+\frac{K^2}{M^2P_c^2}+\frac{1}{MPP_c}+\frac{K}{M^2PP_c^2}+\frac{K}{MP_cQ}+\frac{K^2}{M^2P_c^2Q}+\frac{K}{M^2PP_c^2Q} }.\label{eq:sinr_exp}
\end{equation}
\label{lemma-rate}
\end{lemma}
\begin{proof}
As $\log_2(1+1/x)$ is a convex function of $x$ \cite{convex}, according to Jensen's inequality, we have
\begin{eqnarray}
C_{i}\ge  \frac{1}{2}\log_2\left(1+\frac{1}{\Exp\left\{\frac{1}{{\rm SINR}_i}\right\}}\right).\nonumber
\end{eqnarray}
By applying the SINR approximation in (\ref{eq:SINR_e_app}), we have
\setlength{\arraycolsep}{1pt} 
\begin{eqnarray*}
\frac{1}{\Exp\left\{\frac{1}{{\rm SINR}_i}\right\}}&=&\frac{MP_c^4 }{\Exp\left\lbrace(K-1)P_{i,e}+\frac{1}{P} P_{n,e}+P_{e,1}+P_{e,2}+P_{e,3}+\frac{KP_c^3(1+\frac{K}{MP_c}+\frac{1}{PP_cM})}{Q}\right\rbrace}\\
&=& \frac{1}{\frac{K-1}{M}\left[\frac{2}{P_c}+\frac{K}{MP_c^2}\right]+\frac{1}{MPP_c}+\frac{K}{M^2PP_c^2}+\frac{K}{M^2}(\frac{1}{P_c}-1)^2+\frac{2(1-P_c)}{MP_c}+\frac{K(1+\frac{K}{MP_c}+\frac{1}{PP_cM})}{MP_cQ}},\nonumber\\
&\approx & \frac{1}{\frac{2K}{MP_c}+\frac{K^2}{M^2P_c^2}+\frac{1}{MPP_c}+\frac{K}{M^2PP_c^2}+\frac{K}{MP_cQ}+\frac{K^2}{M^2P_c^2Q}+\frac{K}{M^2PP_c^2Q} }=\widetilde{\rm SINR}_{i}, \end{eqnarray*}
\setlength{\arraycolsep}{5pt}where the approximation is made by ignoring the lower order terms of $M$ when $M\gg 1$. Thus the lower bound in (\ref{rate-LB}) is obtained.
\end{proof}

From (\ref{rate-LB}) and (\ref{eq:sinr_exp}), we can see that the achievable rate lower bound increases logarithmically with $M$ and $P_c$. But its increasing rates with $P$, $Q$, $1/K$ are slower than logarithmic increase. Note that, by using the method in Lemma 1 of \cite{scaling-mimo-qi}, the sum-rate expression in (\ref{rate-LB}) can also be obtained. But with the method in \cite{ scaling-mimo-qi}, the derived expression is an approximation, while our derivations show that it is a lower bound for large $M$. On the other hand, from Lemma 1 of \cite{scaling-mimo-qi}, we know that the lower bound becomes tighter when the number of relay antennas $M$ or the number of sources $K$ increases.  

The parameter $\widetilde{\rm SINR}_{i}$ has the physical meaning of asymptotic effective SINR corresponding to the achievable rate lower bound. Due to the monotonic relationship in (\ref{rate-LB}), to understand the scaling law of the achievable rate is equivalent to understanding the scaling law of $\widetilde{\rm SINR}_{i}$.  

\subsection{Scaling-Law Results} 
Now, the scaling law of the asymptotic effective SINR, $\widetilde{\rm SINR}_{i}$, will be analysed to show how the system performance is affected by the size of the relay antenna array and other network parameters. To have a comprehensive coverage of network setups and applications, for all system parameters including the number of source-destination pairs $K$, the source transmit power $P$, the relay transmit power $Q$, and the CSI quality parameter $P_c$, a general scaling model with respect to $M$ is used. 

Assume that
\begin{equation}
K=\O(M^{r_k}),\quad \frac{1}{P}=\O(M^{r_p}), \quad \frac{1}{Q}=\O(M^{r_q}), \quad  \frac{1}{P_c}=\O(M^{r_c}),
\label{exponents-def}
\end{equation}
where the notation $f(M)=\O\left(g(M)\right)$ means that when $M\rightarrow\infty$, $f(M)$ and $g(M)$ have the same scaling with respect to $M$. In other words, there exists positive constants $C_1,C_2$ and natural number $m$, such that $C_1|g(M)|\le|f(M)|\le C_2|g(M)|$ for all $M\ge m$. Thus the exponents $r_k$, $r_p$, $r_q$, and $r_c$ represents the relative scales of $K$, $1/P$, $1/Q$, and $1/P_c$ with respect to $M$. For practical ranges of the system parameters, we assume that $0\le r_k,r_p,r_q,r_c\le 1$. The reasons are given in the following.
\begin{itemize}
\item The scale of $K$. Following typical applications of massive MIMO, the number of users should increase or keep constant with the number of relay antennas. Thus $r_k\ge 0$. On the other hand,  the number of users  $K$ cannot exceed $M$ since the maximum multiplexing gain provided by the relay antennas is $M$. Thus, $r_k\le 1$. 

\item The scale of $P$ and $Q$. Following the high energy efficiency and low power consumption requirements of massive MIMO, the source and relay transmit power should not increase with the number of relay antennas. But they can decrease as the number of relay antennas increases with the condition that their decreasing rates do not exceed the increasing rate of the antenna number. This is because that the maximum array gain achievable from $M$ antennas is $M$. A higher-than-linear decrease will for sure make the receive SINR a decreasing function of $M$, which contradicts the promise of massive MIMO communications. Thus $0\le r_p,r_q\le 1$.

\item The scale of $P_c$. From the definition of $P_c$ in (\ref{eq:pc}), we have $1/P_c=1+1/E_t$, thus $r_c\ge 0$. This is consistent with the understanding that the CSI quality will not improve as the number of relay antennas increases, as the training process cannot get benefits from extra antennas \cite{massive}. On the other hand, since similar to the data transmission,  the total training  
energy should not has lower scaling than $1/M$, we conclude that $1/P_c$ should not have a higher scaling than $M$. Thus $r_c\le 1$.  
\end{itemize}

In our parameter modelling, the exponents can take any value in the continuous range $[0,1]$. This is different from most existing work where only one or two special values are assumed  for the parameters. Widely used values are 0, 0.5, and 1, which mean that the parameters scale as constant, linear function, and square-root of $M$. Our model covers existing work as special cases.

For the scaling law of  $\widetilde{\rm SINR}_{i}$, denote its scaling with respect to $M$ as 
\begin{equation}
\widetilde{\rm SINR}_{i}=\mathcal{O}(M^{r_s}) \text{, or equivalently, } r_s=\lim_{M\rightarrow\infty} \frac{\log \widetilde{\rm SINR}_{i}}{\log M}.
\label{exp-SINR}
\end{equation} 
The exponent $r_s$ shows the scaling of $\widetilde{\rm SINR}_{i}$.  
\begin{theorem}
For the massive MIMO relay network with MRC/MRT relaying and CSI error, with the model in (\ref{exponents-def}) and (\ref{exp-SINR}), we have the following performance scaling law:
\begin{equation}
r_s=1-r_c-\max(r_p,r_k+r_q).
\label{SNR-scaling}
\end{equation}
\label{thm-1}
\end{theorem}
\vspace{-1cm}
\begin{proof} 
From (\ref{eq:sinr_exp}) we can see that, the maximal scaling exponent of the terms in the denominator determines the scaling exponent of $\widetilde{\rm SINR}_{i}$ with respect to $M$. After some tedious calculation, we find that the term with the highest scaling exponent is either $\frac{1}{MPP_c}$ or $\frac{K}{MP_cQ}$. By using the parameter models in (\ref{exponents-def}), the results in (\ref{SNR-scaling}) is obtained.
\end{proof}
Sensible massive MIMO system should have $r_s\ge0$, i.e., the asymptotic effective SINR and the sum-rate scale as $\O(1)$ or higher. Otherwise, the system performance will decrease with $M$, which contradicts the motivations of massive MIMO systems. To help the presentation, we refer to the case where $r_s\ge0$ as the \textit{favourable-SINR scenario}. The condition for  favourable-SINR  is presented in the following corollary. 
\begin{corollary}
The necessary and sufficient condition for the massive MIMO relay network with MRC/MRT relaying and CSI error to have favourable-SINR is 
\begin{equation}
r_c+\max(r_p,r_k+r_q)\le 1, \quad r_c,r_p,r_q,r_k \in [0,1].
\label{cond-scaling}
\end{equation} 
\label{coro-1}
\end{corollary}
\vspace{-1cm}
\begin{proof}
This is a straightforward extension from (\ref{SNR-scaling}) of Theorem \ref{thm-1}.
\end{proof}

The scaling law in (\ref{SNR-scaling}) illustrates quantitatively the concatenation of the scalings of different parameters and their effects on the network performance. The condition in (\ref{cond-scaling}) forms a region of $r_k$, $r_p$, $r_q$, $r_c$ that makes the SINR favourable. They provide guidelines for the design of the massive MIMO relay network. Next, we discuss the physical meanings of (\ref{SNR-scaling}) and (\ref{cond-scaling}), and several popular network setups. 

Firstly, in (\ref{SNR-scaling}), $r_k$ and $r_q$ appears as a summation. According to their definitions in (\ref{exponents-def}), the summation is the  scaling exponent of $K/Q$. Then in (\ref{SNR-scaling}), $\max(r_p,r_k+r_q)$, which equals $\min(-r_p,-r_k-r_q)$, is the minimum of the power scaling exponents of $P$ and $Q/K$. Recall that $P$ is the per-source transmit power and $Q/K$ is the average relay power allocated to each source. Thus, from (\ref{SNR-scaling}), we can see that the performance scaling of the SINR is determined by two factors: 1) $\max(r_p,r_k+r_q)$, which is the worse per-source power scaling of the two steps, and 2) $P_c$, which is the CSI quality. 

Further, (\ref{SNR-scaling}) shows that $r_s$, which represents the scale of the system SINR,  is a decreasing function of both $\max(r_p,r_k+r_q)$ and $r_c$. Thus high transmit power and better CSI quality result in improved performance. There is a natural tradeoff between the worse per-source power and channel training (e.g., between the data transmission phase and the training phase), and one can compensate for the other in performance scaling. For the two-step communication, the worse step dominates the overall performance.

The condition in (\ref{cond-scaling}) implies $r_k+r_q\le 1$, which means that for the SINR to be favourable, the scaling of the per-source-destination-pair relay power should be no less that $1/M$. This also shows a tradeoff between $r_k$ and $r_q$. Recall that $0\le r_k,r_q\le 1$. That is, with extra relay antennas, we can serve more users or use less relay power for the same level of performance, but the improvement in the two aspects has a total limit. For example, two cases satisfying the constraint are 1) $r_k=1$, $r_q=0$; 2) $r_q=1$, $r_k=0$. The first case means that when the number of users increases linearly with the number of relay antennas (i.e., $r_k=1$), the relay power must remain constant (i.e., $r_q=0$), and thus the goal of saving relay power cannot be achieved. The second case is the opposite: when the relay power is scaled inversely proportional to the number of relay antennas, the goal of serving more users cannot be achieved.

\subsection{Discussions on Several Popular Network Settings}
In this subsection, we further elaborate the scaling law in (\ref{SNR-scaling}) and the condition in (\ref{cond-scaling}) for popular network settings. 
\begin{enumerate}
\item First, we consider the case of $r_c=0$, corresponding to perfect or constant CSI quality case (e.g., $E_t$ increases linearly in $M$). From (\ref{SNR-scaling}) and (\ref{cond-scaling}), the SINR scaling exponent is $r_s=1-\max(r_p,r_k+r_q) $ and the necessary and sufficient condition for favourable SINR is $r_k+r_q\le 1$.
Its physical meaning is that, when the CSI is perfect and for the SINR to be favourable, the most power-saving design is to make both the per-source power of the two hops decrease linearly with the number of antennas. Thus, when the CSI quality is good, we can design the network to serve more users and/or save power consumption, while maintain certain quality-of-service.

\item Next, we consider the case of $r_c=1$, which is equivalent to $E_t= \O(1/M)$. This means that the total energy used in training and the CSI quality are inversely proportional to the relay antenna number. In this case, the SINR scaling exponent is $r_s=-\max(r_p,r_k+r_q).$ To have favourable SINR, from (\ref{cond-scaling}), we need $r_p=r_k=r_q=0$. That is, the source data transmit power, the per-source relay power, and the number of users should all remain constant. This shows that the CSI quality is key to the performance of the massive MIMO relay network. With low CSI quality, all the promising features of the massive MIMO network are gone.  
\item For a general $r_c\in(0,1)$, favourable SINR requires $\max(r_p,r_k+r_q)\le1-r_c$. That is, the worse transmit power per-source of the two steps cannot be lower than $\O(1/M^{1-r_c})$. This shows the trade-off between the training phase and the data transmission phase.
\item For the most power saving setting where $r_p=1$ or $r_k+r_q=1$, the per-source transmit power of either the two steps scales as $1/M$. To have favourable SINR, $r_c=0$ is needed. Thus, the per source transmit power of either or both steps can be made inverse proportional to the number of relay antennas. But at the same time, the CSI quality must remain at least constant, not a decreasing function of $M$. If furthermore $r_k=0$ (e.g., the number of source-destination pairs $K$ remains constant), we have for this setting $P$ or $Q$ scales with $1/M$, which is the major power scaling scenario considered in the literature. It is obvious that our results cover this case, and shows more insights by considering the scales of $K$ and $P_c$.

\item While in previous discussions, $r_c$ is treated as a free parameter, next, we consider the special case of $P_t=P$ and $\tau=K$.   The condition $P_t=P$ corresponds to the practical scenario that user devices always use the same transmit power, no matter for training or data transmission. It is a common assumption in the literature \cite{efficiency-twoway-mrcmrt-csi, efficiency-zf, rate-CSI}. $\tau=K$ is the minimum training length for effective communications \cite{efficiency}. It is shown in \cite{MRCvsZF} that, for maximal-ratio processing, the case achieves the maximal spectral efficiency. We can see that in this case, $r_c=\max\{0, r_p-r_k\}$. Consequently, the SINR scale exponent is  $r_s=1-\max\{0, r_p-r_k\}-\max(r_p,r_k+r_q).$ For the SINR to be favourable, we need $\max(r_k+r_q,2r_p-r_k,r_p+r_q)\le 1$. If further $r_k=0$, i.e., the number of source-destination pairs is constant, favourable SINR requires $r_p\le 1/2$, i.e., the source transmit power can be reduced by $1/\sqrt{M}$ at maximum. This is same as the conclusion as in \cite{efficiency-twoway-mrcmrt-csi, efficiency-zf, rate-CSI}. But note that our model is different from \cite{efficiency-twoway-mrcmrt-csi, efficiency-zf, rate-CSI} and is more general. 

\item Another popular setting is to have the number of source-destination pairs increase linearly with $M$, i.e., $r_k=1$. One example is assuming that $K/M$ is a constant as $M$ increases. From (\ref{SNR-scaling}) and  (\ref{cond-scaling}), for this case, the SINR scaling exponent is $r_s=-r_c-r_{q}$ and to have favourable SINR, we need $r_c=r_{q}=0$. Thus, to support such number of source-destinations, the CSI quality must be high and at the same time the relay power cannot decrease with $M$.
\end{enumerate}



\section{Systems with Asymptotically Deterministic SINR}
\label{sec:deter}
One important concept in massive MIMO systems is asymptotically deterministic. For example, with receiver combining and/or pre-coding at the base station or relay station, random variables such as the signal power and interference power which are random in finite-dimension cases converge to deterministic values as the number of relay antennas is large  \cite{1,2}. This effect is also called channel hardening \cite{massive,massive-1}. With channel hardening, the small-scale fading effect is negligible, and so is the channel variance in the frequency domain. This not only simplifies many design issues but also enables performance analysis via the deterministic equivalences of the random variables, e.g., \cite{efficiency,1,2}. 
One important question is thus when the massive MIMO system have asymptotically deterministic 
SINR for the corresponding performance analysis to be valid. 

In this section, we derive a sufficient condition on asymptotically deterministic SINR and discuss typical scenarios. The result is summarized in the following proposition.

\begin{proposition}
\label{pro:condition}
When $M\gg 1$, a sufficient condition for the SINR to be asymptotically deterministic is 
\begin{eqnarray*}
&&1)\text{ } r_s+r_c+\max\{r_p,r_k+r_q\} = 1,\quad 2)\text{ } 2r_s+2r_c+r_k\le 1, \quad \\ 
&&3)\text{ } 2r_s+3r_c+2r_p \le 2, \quad 4)\text{ } r_c,r_p,r_q,r_k \in [0,1].
\label{eq:suff-con}
\end{eqnarray*}
\end{proposition}
\begin{proof}
Please see Appendix B. 
\end{proof}

From Constraint 2) of the conditions, we can see that $r_s\le 1/2$, meaning that the highest possible SINR scaling is $1/\sqrt{M}$ for the sufficient condition. In addition, $r_c\le 1/2$, meaning that to make the SINR asymptotically deterministic, the CSI quality should scale no lower than $1/\sqrt{M}$. By the definition of $P_c$ in (\ref{eq:pc}), the lowest scaling the training energy $E_t$ can have is $1/\sqrt{M}$. Note that, for a favourable SINR, the scale of the CSI quality parameter can be as low as $1/M$. Therefore, for asymptotically deterministic SINR, the constraint on the CSI quality is more strict.

Next, we investigate typical scenarios for the SINR scaling, which include all possible cases if $r_s$ and $r_c$ are allowed to take values from $\{0,1/2,1\}$ only. 
The tradeoff between parameters will be revealed.
\begin{enumerate}
\item To achieve both $r_s=1/2$ (the SINR increases linearly with $\sqrt{M}$) and asymptotically deterministic SINR, the sufficient condition reduces to $r_k=0$, $r_c=0$, and $\max\{r_p,r_q\}=1/2$. It means that when the number of users and the CSI quality remain constant, the lower of the source power and the relay power must scale as $1/\sqrt{M}$. While in existing work, only constant SINR case ($r_s=0$) has been considered \cite{efficiency, 1,2}, our result shows that the SINR can  scale as $\sqrt{M}$ with asymptotically deterministic property.

\item To achieve $r_s=0$ (constant SINR level) and asymptotically deterministic SINR, two cases may happen: a) $r_c=0$ and $\max\{r_p,r_k+r_q\}=1$; and b) $r_c=1/2$, $r_k=0$, $r_p\le 1/4$ and $r_q= 1/2$.

For Case a), when the CSI quality has constant scaling (e.g., perfect CSI or high quality channel estimation), the scale of the lower per-source transmission power of the two hops should scale as $1/M$. This is the case considered in \cite{1,2}. Similar scenarios for massive MIMO systems without relays have also been reported in \cite{efficiency}. 
Case b) indicates that when the CSI quality scales as $1/\sqrt{M}$ (e,g., the training power scales as $1/\sqrt{M}$ with fixed training length), the number of source-destination pairs should remain constant, the relay power should scale as $1/M$, and the source power can scale smaller than $1/\sqrt[4]{M}$.
\end{enumerate}


\section{Systems with Linearly Increasing SINR}
\label{sec:linear}
In our asymptotically deterministic SINR analysis, the scale of the SINR is no larger than $\O(\sqrt{M})$. While, it can be seen from (\ref{SNR-scaling}) that the maximum scale of the SINR with respect to the number of relay antennas $M$, is $\O(M)$, i.e., linearly increasing with $M$. This is a very attractive scenario for massive MIMO relay networks, in the sense that when $M\gg 1$ significant improvement in the network throughput and communication quality can be achieved. Possible applications for such scenario are networks with high reliability and throughput requirement such as industrial wireless networks and high-definition video. 

In this section, we study networks with linearly increasing SINR. First, the condition on the parameter scaling for the SINR to be linearly increasing is investigated. Then we show that in this case the interference power is not asymptotically deterministic, but with a non-diminishing SCV as $M\rightarrow \infty$. Thus deterministic equivalence analysis does not apply and the small-scale effect needs to be considered in analyzing the performance. We first derive a closed-form PDF of the interference power, then obtain expressions for the outage probability and ABER. Their scalings with network parameters are revealed.  

\begin{proposition}
\label{pro:linearSINR}
When $M\gg 1$, the sufficient and necessary condition for the average SINR to scale as $\O(M)$ is $r_c=r_q=r_p=r_k=0$, i.e., the CSI quality, the source transmit power, the relay power, and the number of users all remain constant. 
In this case, the SINR can be approximated as 
\begin{equation}
{\rm SINR}_{i,e}\approx \frac{M}{P_{i,e}\frac{(K-1)}{P_c^4}+\left(\frac{1}{P}+\frac{K}{Q}\right)(\frac{1}{P_c}+\frac{K}{MP_c^2})+2\left(\frac{1}{P_c}-1\right)+\frac{K}{M}\left(\frac{1}{P_c}-1\right)^2},
\label{eq:SINR_app_high}
\end{equation}
where ${\rm SCV}\{P_{i,e}\}\approx \frac{1}{K-1}$.
\end{proposition}
\begin{proof}
Please see Appendix C.
\end{proof}

Proposition \ref{pro:linearSINR} shows that for linearly-increasing SINR, the interference power is not asymptotically deterministic and does not diminish as $M$ increases. In addition, the randomness of the interference power is the dominant contributor to the random behaviour of the SINR.
With this result, to analyse the outage probability and ABER performance, the distribution of the interference needs to be derived. 

\begin{proposition}
\label{pro:pdf}
Define  
\begin{eqnarray}
&&\rho_e=\frac{1}{\sqrt{M}}\frac{\sqrt{\frac{4}{P_c}+10}}{2+\frac{K}{MP_c}},\label{eq:miu}\\
&&b_e=(K-1)\rho_e, \quad c_e=1-\rho_e ,\quad d_e=\frac{P_c^3}{K-1}\left(2+\frac{K}{MP_c}\right).
\end{eqnarray}
When $M\gg 1$, the PDF of $P_{i,e}$ has the following approximation:
\begin{equation}\label{equ4}
f_{P_{i,e}}(y)=\frac{c_e}{b_e+c_e}\sum\limits_{i=0}^\infty \left(\frac{b_e}{b_e+c_e}\right)^i\phi\hspace{-0.5mm}\left(y;K+i-1,d_ec_e\right),
\end{equation}
where $\phi(y;\alpha,\beta)=\frac{y^{\alpha-1}e^{-y/\beta}}{\beta^\alpha(\alpha-1)!}$ is the PDF of Gamma distribution with shape parameter $\alpha$ and scale $\beta$. It can also be rewritten into the following closed-form expression:
\begin{equation}
\label{cf-pdf-e}
f_{P_{i,e}}(y)\approx\frac{(b_e+c_e)^{K-3}}{d_e
b_e^{K-2}}
\left[e^{-\frac{y}{d_e(b_e+c_e)}}-e^{-\frac{y}{d_ec_e}}
\hspace{-1mm}\sum_{n=0}^{K-3}\hspace{-0.5mm}\frac{1}{n!}\left(\hspace{-0.5mm}\frac{b_e}{d_ec_e(b_e+c_e)}y\right)^{n}\hspace{-0.5mm}\right].\hspace{1cm}
\end{equation}

\end{proposition}
\begin{proof}
Please see Appendix D.
\end{proof}

From (\ref{equ4}), it can be seen that the interference power has a mixture of infinite Gamma distributions with the same scale parameter which is $d_ec_e$ but different shape parameters. But as (\ref{equ4}) is in the form of an infinite summation, it is manipulated into (\ref{cf-pdf-e}) for further analysis. Besides, when the CSI quality is high, i.e., $P_c \approx 1$, we have  $K/(MP_c)\ll 1 $ and thus $\rho_e$ and $d_e$ can be simplified by ignoring the term $K/(MP_c)$. Compared with the perfect CSI case where $P_c=1$, the CSI error makes $d_ec_e$ smaller.


\subsection{Outage Probability Analysis}
Outage probability is the probability that the SINR falls below a certain threshold. Due to the complexity of relay communications, the user-interference, and the large scale, the outage probability analysis of multi-user massive MIMO relay networks is not available in the literature. The derived approximate PDF for the interference power in (\ref{cf-pdf-e}) and the simplified SINR approximation in (\ref{eq:SINR_app_high}) for linearly increasing SINR case allow the following outage probability derivation.

Let $\gamma_{th}$ be the SINR threshold and define 
\[\xi\triangleq \left(\frac{1}{P}+\frac{K}{Q}\right)\left(\frac{1}{P_c}+\frac{K}{MP_c^2}\right)+2\left(\frac{1}{P_c}-1\right)+\frac{K}{M}\left(\frac{1}{P_c}-1\right)^2.
\] The outage probability of User $i$ can be approximated as
\setlength\arraycolsep{1pt}
\begin{eqnarray*}
P_{out}(\gamma_{th})&=&{\mathbb{P}}({\rm SINR}_{i,e}<\gamma_{th})\nonumber \\
&\approx& {\mathbb{P}}\left(\frac{M}{P_{i,e}\frac{K-1}{P_c^4}+\xi}<\gamma_{th}\right)={\mathbb{P}}\left(P_{i,e}>\left(\frac{M}{\gamma_{th}}-\xi\right) \frac{P_c^4}{K-1}\right)\\
&=&\left\{\begin{array}{ll} 1 & \mbox{\ \ if $\gamma_{th}\ge\frac{M}{\xi}$}\\ {\mathbb{P}}\left(P_{i,e}>\left(\frac{M}{\gamma_{th}}-\xi\right) \frac{P_c^4}{K-1} \right)& \mbox{\ \ otherwise}\end{array}\right..\label{outage_expression}
\end{eqnarray*}
When $\gamma_{th}< \frac{M}{\xi}$, from (\ref{cf-pdf-e}), we have 
\begin{eqnarray}
&&\hspace{-2mm} P_{out}(\gamma_{th}) \approx \left(\frac{b_e}{b_e+c_e}\right)^{\hspace{-1mm} 2-K}\hspace{-2mm} e^{-\frac{\left(\frac{M}{\gamma_{th}}-\xi\right)P_c^4}{(K-1)d_e(b_e+c_e)}}\nonumber \\
&&\qquad-\frac{c_e}{b_e+c_e} \hspace{-1mm}\sum_{n=0}^{K-3}
\frac{1}{\Gamma(n+1)} \left(\frac{b_e}{b_e+c_e}\right)^{n-K+2}\hspace{-2mm}\Gamma\hspace{-1mm}\left(n+1,\frac{\left(\frac{M}{\gamma_{th}}-\xi\right)P_c^4 }{(K-1)d_ec_e}\right)\hspace{-1mm},
\label{outageprob}
\end{eqnarray}
 where $\Gamma(s,x)\triangleq\int_x^\infty t^{s-1}e^{-t}\mathrm{d}t$ is the upper incomplete gamma function \cite{Int}. This outage probability expression is too complex for useful insights. A simplified one is derived in the following proposition for systems with high CSI quality.
\begin{proposition} \label{pro:outage_app}
Define $$D\triangleq \frac{\left(2\left(1-P_c\right)+\frac{1}{P}+\frac{K}{Q}\right)P_c^3}{(K-1)d_e(b_e+c_e)}.$$
When $E_t\gg1$ and $M\gg \gamma_{th}\left(2d_ec_e(1+\frac{c_e}{{b_e}}) K(K-1)+\frac{1}{P}+\frac{K}{Q}\right)$,  we have
\begin{equation}
P_{out}(\gamma_{th}) \approx \left(\frac{b_e}{b_e+c_e}\right)^{2-K}e^{D-\frac{MP_c^4}{\gamma_{th}(K-1)d_e(b_e+c_e)}}.
\label{outate_app}
\end{equation}
\end{proposition}
\begin{proof} 
By the definitions of $P_c$ and $E_t$ in (\ref{eq:pc}), when $E_t\gg1$, we have $P_c\approx 1$. Thus $\xi\approx 1/P+K/Q$. Further define  $$a\triangleq\frac{b_e\left(\frac{M}{\gamma_{th}}-\xi\right)P_c^4}{(K-1)d_ec_e(b_e+c_e)}.$$  
When $M\gg \gamma_{th}\left(2d_ec_e(1+\frac{c_e}{{b_e}}) K(K-1)+\frac{1}{P}+\frac{K}{Q}\right)$, we have $a\gg 2K>1$ and therefore
 $$\frac{\left(\frac{M}{\gamma_{th}}-\xi\right)P_c^4}{(K-1)d_ec_e} \gg 1.$$
  Then, from \cite[8.357]{Int} we know that
$$\Gamma\left(n+1,\frac{\left(\frac{M}{\gamma_{th}}-\xi\right)P_c^4 }{(K-1)d_ec_e}\right)\approx
\left(\frac{\left(\frac{M}{\gamma_{th}}-\xi \right)P_c^4}{(K-1)d_ec_e}\right)^n e^{-\frac{\left(\frac{M}{\gamma_{th}}-\xi\right)P_c^4}{(K-1)d_ec_e}}.$$
 With this approximation, the outage probability expression in (\ref{outageprob}) can be reformulated as
\begin{eqnarray*}
\hspace{-2mm} P_{out}(\gamma_{th}) &\approx & \left(\hspace{-1mm}\frac{b_e}{b_e+c_e}\hspace{-1mm}\right)^{\hspace{-1mm}2-K}\hspace{-4mm} e^{-\frac{\left(\frac{M}{\gamma_{th}}-\xi\right)P_c^4}{(K-1)d_e(b_e+c_e)}}\hspace{-2mm}\left[1-\frac{c_e}{b_e+c_e}e^{-a}\sum_{n=0}^{K-3}\frac{a^n}{\Gamma(n+1)}\right]
\end{eqnarray*}



Notice that as $a\gg 2K>1$, $e^{a}\gg\sum_{n=0}^{K-3}\frac{a^n}{\Gamma(n+1)}$. Thus the second term in the bracket of the previous formula can be ignored, and the approximation in (\ref{outate_app}) is obtained.
\end{proof}

We can see that the outage probability approximation in (\ref{outate_app}) is tight when the number of relay antennas is much larger than the number of source-destination pairs and the training power and transmit powers are high. These conditions will result in a high received SINR. Thus, the approximation in (\ref{outate_app}) applies to the high SINR case.

Note that (\ref{outate_app}) can also be obtained by deleting the second summation term in the PDF formula in (\ref{cf-pdf-e}) and then integrating with the approximated PDF. This is because that, for the high SINR case, the outage probability is determined by the SINR distribution in the small SINR region, which is equivalently the high interference power region, corresponding to the tail of the PDF of the interference power. It can be seen from the PDF in (\ref{cf-pdf-e}) that, the first term has a heavier tail, thus dominates the outage probability.

Now, we explore insights from (\ref{outate_app}). As $b_e,c_e,d_e$ are independent with $P$ or $Q$, the outage probability scales as $e^{\frac{P_c^3}{P(K-1)d_e(b_e+c_e)}}$ with $P$ and scales as $e^{\frac{KP_c^3}{Q(K-1)d_e(b_e+c_e)}}$ with $Q$. Firstly, it shows the natural phenomenon that increasing $P$ or $Q$ will decrease the outage probability. Also, we can see that the outage probability curve with respect to $Q$ has a sharper slope than that with $P$. For example, let $P=Q=\alpha$, doubling $P$ alone will shrink the outage probability by a factor of $e^{\frac{P_c^3}{2(K-1)d_e(b_e+c_e)\alpha}}$, while doubling $Q$ alone will shrink the outage probability by a factor of $e^{\frac{KP_c^3}{2(K-1)d_e(b_e+c_e)\alpha}}$, which is $K$ powers of the shrinkage of the doubling-$P$ case. Furthermore,
the outage probability will not diminish to zero as the user and relay transmit power increase. An error floor exists due to the user-interference. On the other hand, increasing the number  of relay antennas to infinity leads to faster decrease in the outage probability and makes it approach zero. 

Note that in our analysis, we assume $M\gg 1$ but does not go to infinity. So terms with $1/\sqrt{M}$ are not treated as asymptotically small and thus are not ignored. If $M\rightarrow \infty$ and $P_c\rightarrow 1$, the $1/\sqrt{M}$ terms can be seen as $0$ and we will have $P_{out}(\gamma_{th})\approx\left(\frac{(K-1)\sqrt{3.5}}{\sqrt{M}}\right)^{2-K}e^{-\frac{M}{2\gamma_{th}}}.$ However, this asymptotic analysis is not practical because the number of massive MIMO antennas is usually a few hundreds in practice, so that $\sqrt{M}$ may not be much larger than other parameters such as $K,P,Q$.

\subsection{ABER analysis}
ABER is anther important performance metric. Due to the complexity of the SINR distribution, ABER analysis of the massive MIMO relay network is not available in the literature. For the linearly increasing SINR case, the ABER can be analyzed as below. 

Denote the ABER as $P_b(e)$. It is given by
\begin{equation}
P_b(e)=\int_0^{\infty}P_b(e|r)f_{\rm SINR}(r){\rm d}r,\label{eq:defpb}
\end{equation}
where $P_b(e|r)$ is the conditional error probability and $f_{\rm SINR}(r)$ is the PDF of the SINR. For  channels with additive white Gaussian noise, $P_b(e|r)=A {\rm erfc}\left(\sqrt{B r}\right)$ for several Gray bit-mapped constellations employed in practical systems, where ${\rm erfc}(x)$ is the complementary error function, $A$ and $B$ are constants depended on the modulation. For example, for BPSK, $A=0.5$, $B=1$. 

For the linearly increasing SINR case, With the PDF of the interference power in (\ref{cf-pdf-e}) and the SINR approximation in (\ref{eq:SINR_app_high}), the PDF of the SINR can be derived as below.
\begin{eqnarray}
&&f_{\rm SINR}(r)=\frac{\left(b_e+c_e\right)^{K-3}MP_c^4}{r^2 (K-1)d_eb_e^{K-2}}e^{-\frac{\left(\frac{M}{r}-\xi\right)P_c^4}{(K-1)d_e(b_e+c_e)}}\nonumber\\
&&\hspace{1.9cm}-\sum_{n=0}^{K-3}\frac{(b_e+c_e)^{K-n-3}MP_c^{4n+4}}{\Gamma(n+1)((K-1)d_e)^{n+1}c_e^nb_e^{K-n-2}}\frac{\left(\frac{M}{r}-\xi\right)^n}{r^2}e^{-\frac{\left(\frac{M}{r}-\xi\right)P_c^4}{(K-1)d_ec_e}}, r\in\left(0,\frac{M}{\xi}\right).  \label{pdf_SINR}
\end{eqnarray}
By using (\ref{pdf_SINR}) in (\ref{eq:defpb}), an approximation on the ABER is derived in the following proposition. 

\begin{proposition}
\label{pro:BER}
When $E_t\gg1$ and $M\gg 2d_ec_e(1+c_e/b_e)K(K-1)+\frac{1}{P}+\frac{K}{Q}$, the ABER can be approximated as
\begin{eqnarray}
&&P_b(e)\approx A\left(\frac{b_e}{b_e+c_e}\right)^{2-K}e^{D-2P_c^2\sqrt{\frac{BM}{(K-1)d_e(b_e+c_e)}}}. \label{eq:BER_2}
\end{eqnarray}
\end{proposition} 

 \begin{proof}
The PDF of the SINR in (\ref{pdf_SINR}) can be rewritten as 
\begin{eqnarray}
&&\hspace{-2cm} f_{\rm SINR}(r)=\frac{\left(b_e+c_e\right)^{K-3}MP_c^4}{r^2 (K-1)d_eb_e^{K-2}}e^{-\frac{\left(\frac{M}{r}-\xi\right)P_c^4}{(K-1)d_e(b_e+c_e)}}\left[1-\frac{\sum_{n=0}^{K-3}\frac{\left(\frac{b_e\left(\frac{M}{r}-\xi\right)P_c^4}{(K-1)d_ec_e(b_e+c_e)}\right)^n}{\Gamma(n+1)}}{e^{\frac{b_e\left(\frac{M}{r}-\xi\right)P_c^4}{(K-1)d_ec_e(b_e+c_e)}}}\right].\label{pdf_SINR_1}
\end{eqnarray}
As the ABER is determined by the PDF when $r$ is small \cite{giannakis}, we consider the range $r<1$. With $E_t\gg1$ and $M\gg 2d_ec_e(1+c_e/b_e)K(K-1)+1/P+K/Q$, similarly as the proof of Proposition \ref{pro:outage_app}, we can show that ${\sum_{n=0}^{K-3}{\left(\frac{b_e\left(\frac{M}{r}-\xi\right)P_c^4}{(K-1)d_ec_e(b_e+c_e)}\right)^n}/{\Gamma(n+1)}}{/}{e^{\frac{b_e\left(\frac{M}{r}-\xi\right)P_c^4}{(K-1)d_ec_e(b_e+c_e)}}}\ll 1$, and thus this term can be ignored. 
%

The ABER can be derived by solving $\int_{r=0}^{M/\xi}A {\rm erfc}(\sqrt{Br})f_{\rm SINR}(r) dr$. As the ABER is determined by the region when $r$ is small, we replace the integration region with $\int_{r=0}^{\infty}$ for a tractable approximation. By using ${\rm erfc}(x)=\Gamma(\frac{1}{2},x^2)/\sqrt{\pi}$, the integration formula  $\int_{0}^{\infty}e^{-\mu x}\Gamma(v,\frac{a}{x}){\rm d}x=2a^{v/2}\mu^{v/2-1}K_v(2\sqrt{\mu a})$ \cite{Int}, and $K_{\frac{1}{2}}(x)=\sqrt{\frac{\pi}{2x}}e^{-x}$ \cite{Handbook}, the ABER approximation in (\ref{eq:BER_2}) is obtained.
\end{proof} 

 We can see from (\ref{eq:BER_2}) that increasing $M$ will make the ABER decrease and approach zero. Besides, for very large $M$ the ABER behaves as  $Ce^{-C'\sqrt{M}}$.  As is known, the ABER of traditional MIMO system with $M$ transmit antennas and 1 receive antenna under Rayleigh fading is $C_{MIMO}{\rm SINR}^{-C_{MIMO}'M}$. This shows different ABER behaviour in the massive MIMO relay network, where the ABER decreases exponentially with respect to $\sqrt{M}$. If the diversity gain definition of traditional MIMO system is used \cite{MIMO}, the massive relay network will have infinite diversity gain.
 
 Comparing (\ref{eq:BER_2}) with (\ref{outate_app}), we see that the ABER and the outage probability has the same scaling with $P$ and $Q$ respectively. Thus $P$, $Q$ scaling analysis for the outage probability also applies to the ABER. In addition, if the threshold is set as $\gamma_{th}=\sqrt{\frac{MP_c^4}{4B(K-1)d_e(b_e+c_3)}}$, the ABER equals $A$ times the outage probability. Thus, there is a simple transformation between the two metrics.

\section{Simulation Results}
\label{sec:simu}

In this section, simulation results are shown to verify the analytical results. 
\begin{table}
\caption{The network settings for Figure 1}
\center
\begin{tabular}{|c|c|c|c|c|c|}
\hline
          & $P_c$ & $P$ & $Q$ & $K$ & $r_s$\\
\hline
Case 1& 0.8 & 10 & 10 & $M/10$ & 0\\
\hline   
Case 2& $100/M$& 10 & 10& 10 & 0\\
\hline
Case 3& 0.8 & 10 & $1/\sqrt{M}$& $\lfloor\sqrt{M}\rfloor$ &0\\
\hline
Case 4& 0.8 & 1 & 1& 20 &  1\\
\hline
Case 5& $10/\sqrt{M}$& 10 &10 &20& $1/2$\\
\hline
\end{tabular}
\label{table}
\end{table}
\begin{figure}
\center
\includegraphics[width=5in]{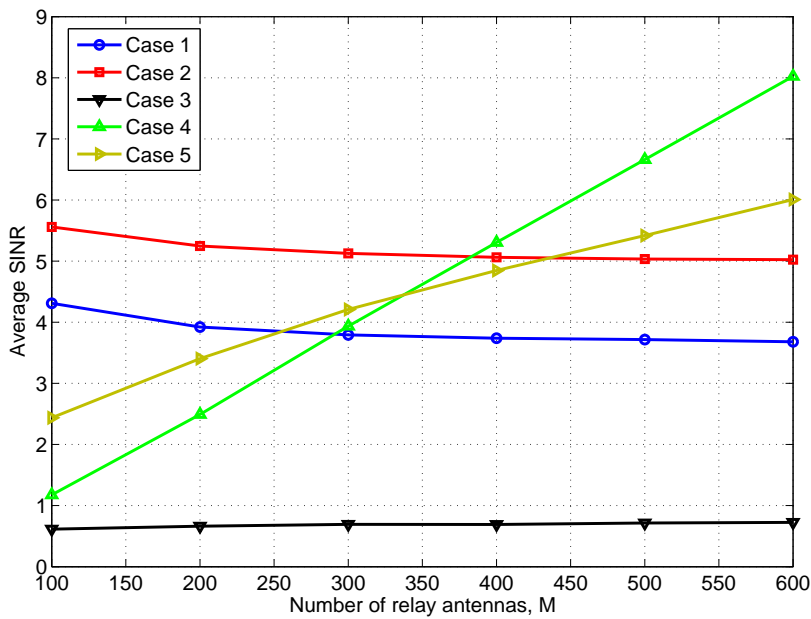}\vspace{-5mm}
\caption{Average SINR v.s.~the number of relay antennas $M$ for different network scenarios.}
\label{fig:scaling}
\end{figure}
In Fig.~\ref{fig:scaling}, the simulated average SINR with respect to the number of relay antennas $M$ is shown for the five network settings given in Table \ref{table} to verify the SINR scaling result in Theorem \ref{thm-1}. In the table, $\lfloor \cdot \rfloor$ is the floor function. For different settings of network parameters, their SINR scalings ($r_s$ values) are calculated based on the SINR scaling law in (\ref{SNR-scaling}) and shown in the table. The first three cases have constant scaling. In Case 4 and Case 5, the average SINR scale as $\mathcal{O}(M)$ and $\mathcal{O}(\sqrt{M})$. The figure verifies these scaling law results.

\begin{figure}[t]
\center
\includegraphics[width=5in]{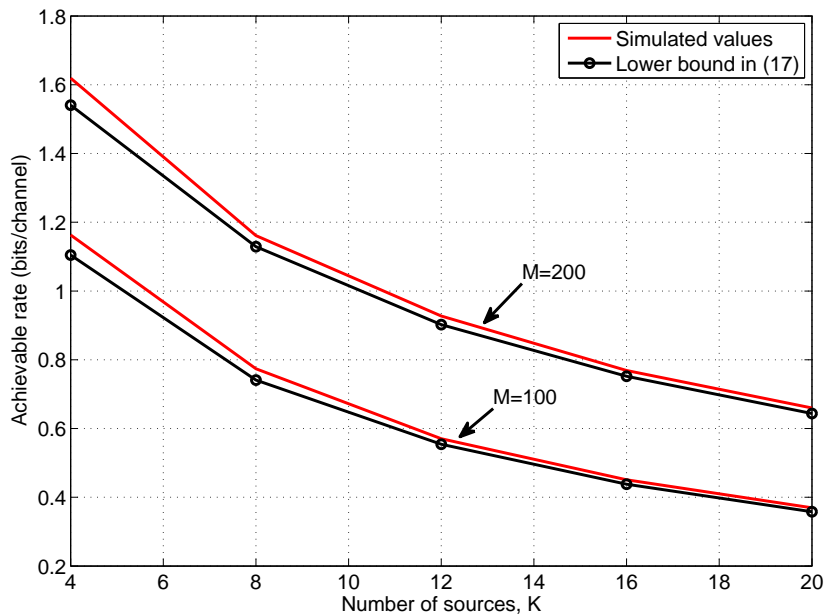} \vspace{-5mm}
\caption{Achievable rate for different number of sources. $M=200$  and $100$, $P=Q=0$ dB, $P_c=1/2$.}
\label{fig-rate}
\end{figure}
In Fig.~\ref{fig-rate}, the average achievable rate per source-destination pair is simulated for different number of sources with $200$ or $100$ relay antennas. The source and the relay powers are set to be $0$ dB. The CSI quality is set as $P_c=1/2$. We can see that the lower bound in (\ref{rate-LB}) is very tight. With given number of relay antennas, the achievable rate per source-destination pair decreases as there are more pairs.

\begin{figure}[t]
\center
\includegraphics[width=5in]{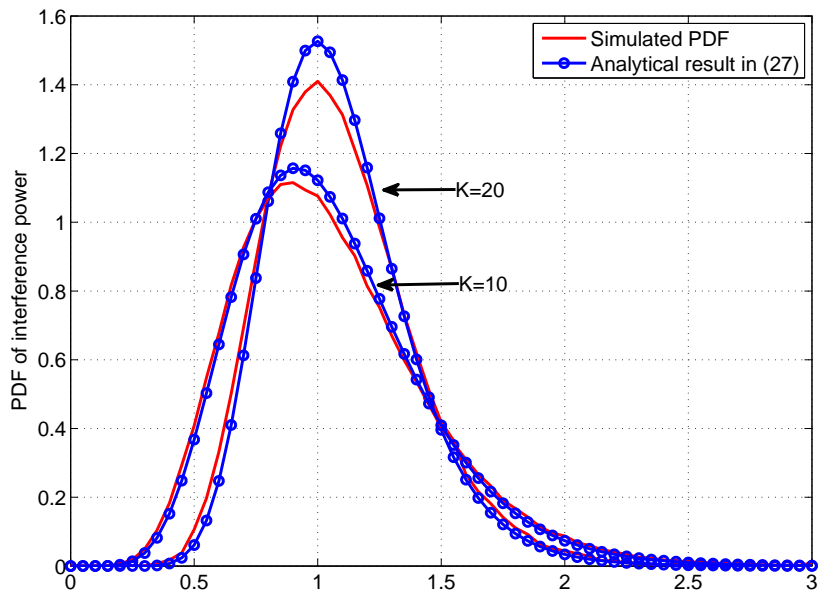} \vspace{-5mm}
\caption{PDF of interference power. $K=20$ or $10$, $P_c=0.8$, $M=200$.}
\label{PDF_inter}
\end{figure}
In Fig.~\ref{PDF_inter}, for a relay network with $20$ or $10$ source-destination pairs and $200$ relay antennas, the simulated PDF of $P_{i,e}$ is shown. The CSI quality parameter is set as $P_c=0.8$. The analytical expression in (\ref{cf-pdf-e}) is compared with the simulated values. We can see from Fig.~\ref{PDF_inter} that the PDF approximation is tight for the whole parameter range. Especially, the approximation matches tightly at the tail when the interference power is large, which is the dominate range of outage and ABER.

\begin{figure}[t]
\center
\includegraphics[width=5in]{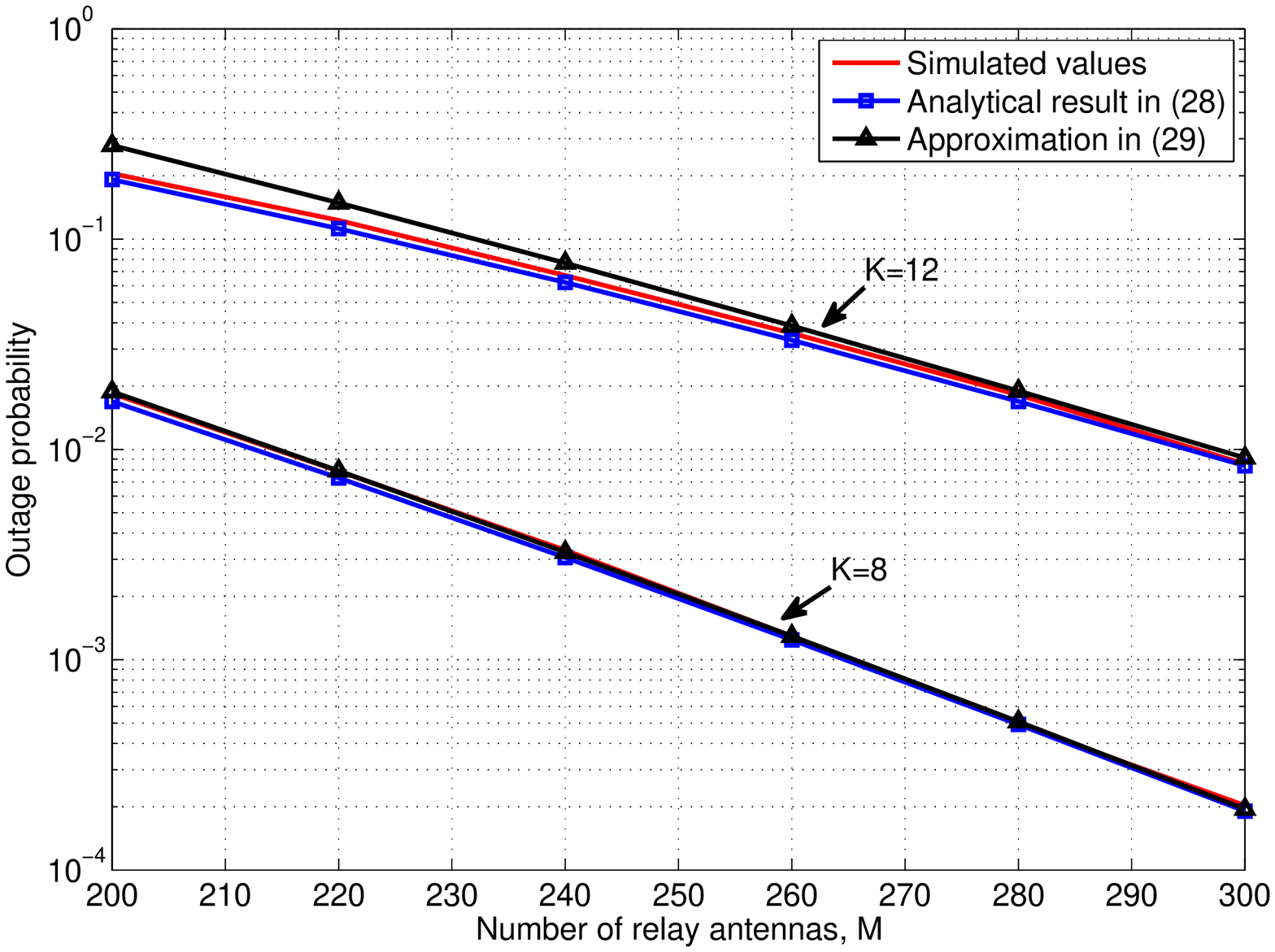} \vspace{-5mm}
\caption{Outage probability for different number of relay antennas. $K=8 \text{ or } 12, P=Q=10$ dB, $\gamma_{th}=8$dB, $P_c=0.95$.}
\end{figure}
Fig.~4 shows the outage probability for different number of relay antennas. The analytical expressions in (\ref{outageprob}) and (\ref{outate_app}) are compared with the simulated values.  The transmit powers of the users and the relay are set as $10$ dB. The CSI quality parameter is set as $P_c=0.95$. The number of sources is $8$ or $12$ and the SINR threshold is $8$ dB. We can see that our analytical result in (\ref{outageprob}) and the further approximation in (\ref{outate_app}) are both tight for all the simulated parameter ranges. Besides, the approximations becomes tighter as the relay antennas number increases.

\begin{figure}[t]
\center
\includegraphics[width=5in]{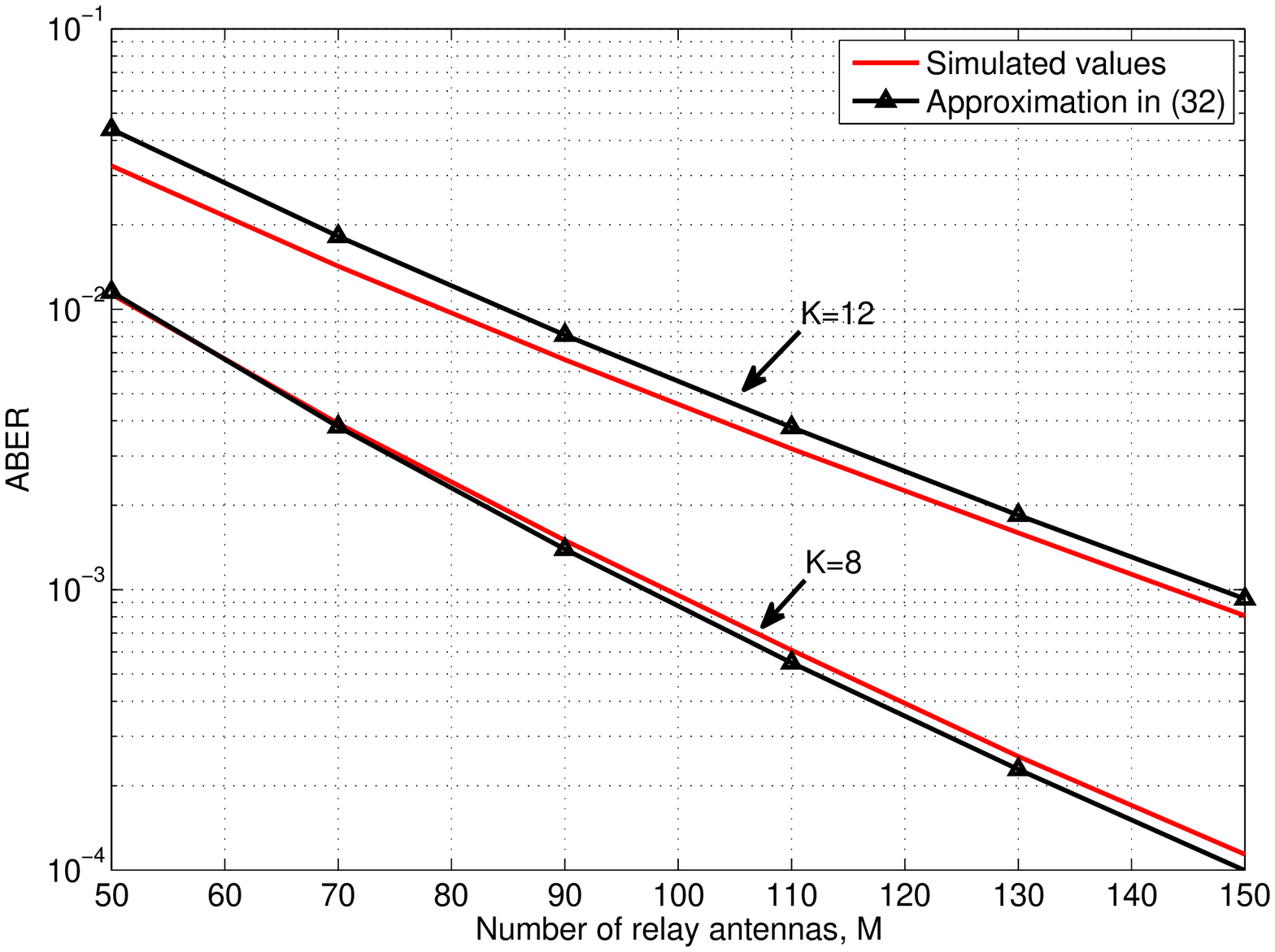} \vspace{-5mm}
\caption{Average bit error rate of BPSK for different number of relay antennas $M$. $K=8 \text{or} 12, P=Q=10$ dB, $P_c=0.95$.}
\label{fig:aber_M_Ksmall}
\end{figure}
In Fig.~\ref{fig:aber_M_Ksmall}, the ABER for BPSK is simulated for different number of relay antennas with $K=8$ or $12$, $P=Q=10$ dB and $P_c=0.95$. The analytical approximation in (\ref{eq:BER_2}) is compared with the simulated values. From the figure, we can see that the analytical result in (\ref{eq:BER_2}) is tight for the simulated values, and is tighter when the number of source-destination pairs is smaller.

\section{Conclusion}
\label{sec:con}
In this work, we analysed the performance of a massive MIMO relay network with multiple source-destination pairs under MRC/MRT relaying with imperfect CSI. Firstly, the performance scaling law is analysed which shows that the scale of the SINR is decided by the summation of the scales of the CSI quality plus the larger of the per-source transmission power of the two hops. With this result, typical scenarios and trade-off between parameters are shown. Our scaling law is comprehensive as it takes into considerations many network parameters, including the number of relay antennas, the number of source-destination pairs, the source transmit power and the relay transmit power. Then, a sufficient condition for asymptotically deterministic SINR is derived, based on which new network scenarios for systems with the asymptotically deterministic property are found
and tradeoff between the parameters is analysed.  At last, we specify the necessary and sufficient condition for networks whose SINR increases linearly with the number of relay antennas. In addition, our work show that for this case the interference power does not become asymptotically deterministic and derived the PDF of the interference power in closed-form. Then the outage probability and ABER expressions for the relay network are obtained and their behaviour with respect to network parameters are analysed. Simulations show that the analytical results are tight.

\section*{Appendix}

\subsection{Derivations of $\Exp\{P_{s,e}\}$ and ${\rm SCV}\{P_{s,e}\}$}
\label{sec-app1}
Firstly, we have
\begin{eqnarray}
\Exp\{P_{s,e}\}&=&\Exp\left\lbrace\frac{|\hat{\bf g}_i\hat{\bf g}_i^H\hat{\bf f}_i^H\hat{\bf f}_i+\sum_{k=1,k\neq i}^{K}\hat{\bf g}_i\hat{\bf g}_k^H\hat{\bf f}_k^H\hat{\bf f}_i|^2}{M^4}\right\rbrace \nonumber\\
&=&\Exp\left\lbrace\frac{\left(\|\hat{\bf g}_i\|_F^2\|\hat{\bf f}_i\|_F^2+\sum_{k=1,k\neq i}^{K}\hat{\bf g}_i\hat{\bf g}_k^H\hat{\bf f}_k^H\hat{\bf f}_i\right)\left(\|\hat{\bf g}_i\|_F^2\|\hat{\bf f}_i\|_F^2+\sum_{k=1,k\neq i}^{K}\hat{\bf g}_i\hat{\bf g}_k^H\hat{\bf f}_k^H\hat{\bf f}_i\right)^H}{M^4}\right\rbrace\nonumber\\
&=&\Exp\left\lbrace\frac{\|\hat{\bf g}_i\|_F^4\|\hat{\bf f}_i\|_F^4}{M^4}\right\rbrace+ \sum_{k=1,k\neq i}^{K}\Exp\left\lbrace\frac{|\hat{\bf g}_i\hat{\bf g}_k^H\hat{\bf f}_k^H\hat{\bf f}_i|^2}{M^4}\right\rbrace,\label{appen_1}
\end{eqnarray}
where the last step is obtained because the means of the cross terms are zero.

In the first term of (\ref{appen_1}), as entries of $\hat{\bf g}_i$ and $\hat{\bf f}_i$ are i.i.d. whose distribution follows $\mathcal{CN}(0, P_c)$, $\|\hat{\bf g}_i\|_F^2$ and $\|\hat{\bf f}_i\|_F^2$ have a gamma distribution with shape parameter $M$ and scale parameter $P_c$. Thus, 
$$\Exp\left\lbrace \frac{\|\hat{\bf g}_i\|_F^4\|\hat{\bf f}_i\|_F^4}{M^4}\right\rbrace=P_c^4\left(1+\frac{2}{M}+\frac{1}{M^2}\right)\approx P_c^4,$$
where the approximation is by ignoring lower order terms of $M$ when $M\gg 1$. For the remaining terms, 
$$\hat{\bf g}_i\hat{\bf g}_k^H\hat{\bf f}_k^H\hat{\bf f}_i=\sum_{m_g=1}^M\sum_{m_f=1}^M\hat{ g}_{i,m_g}\hat{ g}_{k,m_g}^*\hat{ f}_{k,m_f}^*\hat{ f}_{i,m_f},$$
where $\hat{ g}_{i,m_g}$ is the $(i,m_g)$th entry of $\hat{\bf G}$, and $\hat{ f}_{i,m_f}$ is the $(i, m_f)$th entry of $\hat{\bf F}$. Thus $\hat{\bf g}_i\hat{\bf g}_k^H\hat{\bf f}_k^H\hat{\bf f}_i$ can be seen the summation of $M^2$ terms of i.i.d. random variables, each with mean $0$, variance $P_c^2$. According to CLT, the distribution of $\frac{\hat{\bf g}_i\hat{\bf g}_k^H\hat{\bf f}_k^H\hat{\bf f}_i}{M}$ converges to $\mathcal{CN}(0,P_c^4)$ when $M\rightarrow\infty$. Then $\frac{|\hat{\bf g}_i\hat{\bf g}_k^H\hat{\bf f}_k^H\hat{\bf f}_i|^2}{M^4}$ has a gamma distribution with shape parameter $1$ and scale parameter $P_c^4/M^2$. Thus, we can obtain 
$$\sum_{k=1,k\neq i}^{K}\Exp\left\lbrace\frac{|\hat{\bf g}_i\hat{\bf g}_k^H\hat{\bf f}_k^H\hat{\bf f}_i|^2}{M^4}\right\rbrace=\frac{(K-1)P_c^4}{M^2}.$$
As $M\gg K$, we have $\frac{(K-1)P_c^4}{M^2}\ll P_c^4$. Thus the mean of $P_s$ is $P_c^4$. 

Similarly, we can derive the variance of $P_{s,e}$ as below. 
\begin{eqnarray*}
{\rm Var}\{P_{s,e}\}&=&\Exp\left\lbrace \frac{|\hat{\bf g}_i\hat{\bf g}_i^H\hat{\bf f}_i^H\hat{\bf f}_i+\sum_{k=1,k\neq i}^{K}\hat{\bf g}_i\hat{\bf g}_k^H\hat{\bf f}_k^H\hat{\bf f}_i|^4}{M^8} \right\rbrace-\Exp\{P_{s,e}\}^2\\
&\approx & \Exp\left\lbrace \frac{\|\hat{\bf g}_i\|_F^8\|\hat{\bf f}_i\|_F^8}{M^8}\right\rbrace-P_c^8\left(1+\frac{2}{M}+\frac{1}{M^2}
\right)^2\\
&=& P_c^8\frac{(M+3)^2(M+2)^2(M+1)^2}{M^6}-P_c^8\left(1+\frac{2}{M}+\frac{1}{M^2}
\right)^2\approx \frac{8P_c^8}{M}.
\end{eqnarray*}
 
Then, we have ${\rm SCV}\{P_{s,e}\}={\rm Var}\{P_{s,e}\}/(\Exp\{P_{s,e}\})^2=8/M$.

\subsection{Proof of Proposition \ref{pro:condition}}
\label{app-B}
The SINR expression in (\ref{eq:SINR_e}) can be reformulated as 
\begin{equation}
{\rm SINR}_{i,e}= M^{r_s}\frac{P_{s,e}/P_c^4 }{P_{i,e}\frac{K-1}{P_c^4M^{1-r_s}}+\frac{1}{PP_c^4M^{1-r_s}}P_{n,e}+\frac{1}{P_c^4M^{1-r_s}}(
P_{e,1}+P_{e,2}+P_{e,3})+\frac{K(1+\frac{K}{MP_c}+\frac{1}{PP_cM})}{QP_cM^{1-r_s}}}.\label{eq:asym}
\end{equation}

The received SINR is asymptotically deterministic when its SCV approaches zero as $M\rightarrow \infty$. However, due to the complex structure of the SINR expression, it is highly challenging to obtain its SCV directly. Alternatively, as is shown in Section III, $P_{s,e}/P_c^4$ is asymptotically deterministic, thus for the SINR to be asymptotically deterministic, the sufficient and necessary condition is that the denominator of the formula in (\ref{eq:asym}) is asymptotically deterministic. One sufficient condition is that the SCV of the denominator denoted as ${\rm SCV}_d$, is no larger than $E/M$ for some constant $E$\footnote{Note that, when $M\rightarrow\infty$, given any positive number $\alpha$, $1/M^{\alpha}\rightarrow 0$. But for practical applications of the deterministic equivalence analysis in large but finite-dimension systems, we consider the scenario that the SCV decrease linearly with the number of antennas or faster. The derived condition is thus sufficient but not necessary.}. This can be expressed as
\begin{eqnarray}
{\rm SCV}_{d}&=&\frac{{\rm Var}\left\lbrace  P_{i,e}\frac{K-1}{P_c^4M^{1-r_s}}+\frac{1}{PP_c^4M^{1-r_s}}P_{n,e}+\frac{1}{P_c^4M^{1-r_s}}(
P_{e,1}+P_{e,2}+P_{e,3})\right\rbrace}{\left(\Exp\left\lbrace  P_{i,e}\frac{K-1}{P_c^4M^{1-r_s}}+\frac{1}{PP_c^4M^{1-r_s}}P_{n,e}+\frac{1}{P_c^4M^{1-r_s}}(
P_{e,1}+P_{e,2}+P_{e,3})\right\rbrace\right)^2}\le\frac{E}{M}. \label{eq:scv}
\end{eqnarray}
From (\ref{eq:asym}), we have $$\frac{P_{s,e}/P_c^4 }{P_{i,e}\frac{K-1}{P_c^4M^{1-r_s}}+\frac{1}{PP_c^4M^{1-r_s}}P_{n,e}+\frac{1}{P_c^4M^{1-r_s}}(
P_{e,1}+P_{e,2}+P_{e,3})+\frac{K(1+\frac{K}{MP_c}+\frac{1}{PP_cM})}{QP_cM^{1-r_s}}}=\O(1)$$ and since $P_{s,e}/P_c^4 \overset{m.s.}{\longrightarrow} 1$, we have $$\Exp\left\lbrace  P_{i,e}\frac{K-1}{P_c^4M^{1-r_s}}+\frac{1}{PP_c^4M^{1-r_s}}P_{n,e}+\frac{1}{P_c^4M^{1-r_s}}(
P_{e,1}+P_{e,2}+P_{e,3})\right\rbrace=\O(1).$$ Thus (\ref{eq:scv}) is equivalent to that
\begin{equation}
{\rm Var}\left\lbrace  P_{i,e}\frac{K-1}{P_c^4M^{1-r_s}}+\frac{1}{PP_c^4M^{1-r_s}}P_{n,e}+\frac{1}{P_c^4M^{1-r_s}}(
P_{e,1}+P_{e,2}+P_{e,3})\right\rbrace \le \frac{E'}{M}
\label{eq:var_scv}
\end{equation}
for some constant $E'$.
\begin{lemma}
\label{lemma:var}
A sufficient condition for (\ref{eq:var_scv}) is that the variance of each term in (\ref{eq:var_scv}) scales no larger than $1/M$, i.e., the maximum scale order of 
$
 {\rm Var}\left\lbrace P_{i,e}\frac{K-1}{P_c^4M^{1-r_s}}\right\rbrace$, 
${\rm Var}\left\lbrace\frac{1}{PP_c^4M^{1-r_s}}P_{n,e}\right\rbrace$, 
${\rm Var}\left\lbrace \frac{1}{P_c^4M^{1-r_s}}P_{e,1}\right\rbrace$, 
${\rm Var}\left\lbrace \frac{1}{P_c^4M^{1-r_s}}P_{e,2}\right\rbrace$, and
${\rm Var}\left\lbrace \frac{1}{P_c^4M^{1-r_s}}P_{e,3}\right\rbrace$ is no larger than $1/M$.
\end{lemma}
\begin{proof}
The variance of $P_{i,e}\frac{K-1}{P_c^4M^{1-r_s}}+\frac{1}{PP_c^4M^{1-r_s}}P_{n,e}+\frac{1}{P_c^4M^{1-r_s}}(
P_{e,1}+P_{e,2}+P_{e,3})$ is the summation of two parts: the variances of each term, and the covariance of every two terms. 
Now, we will prove that if the variances of each term scales no larger than $1/M$, their covariance also scales no larger than $1/M$.

To make it general and clear, we define $Y=\sum_{n=1}^N X_n$, where $N$ is a finite integer and $X_n$'s are random variables. Without loss of generality, we assume that ${\rm Var}\{X_1\}$ has the highest scale among all ${\rm Var}\{X_n\}$'s and ${\rm Var}\{X_1\}=\O(1/M^\alpha)$, where $\alpha\ge 1$. The variance of $Y$ is
$${\rm Var}\{Y\}=\sum_{n=1}^N{\rm Var}\{X_n\}+\sum_{i\neq j}{\rm Cov}\{X_i,X_j\}.$$
By the definition of covariance, $\sum_{i\neq j}{\rm Cov}\{X_i,X_j\}$ takes the maximum value when $X_n$'s are linearly correlated, i.e., $X_1=X_2/a_2=X_3/a_3\dots=X_N/a_N$. In this case, we can obtain that $$\sum_{i\neq j}{\rm Cov}\{X_i,X_j\}={\rm Var} \{X_1\}\sum_{i\neq j}a_ia_j,$$
where we have defined $a_1=1$.

As ${\rm Var}\{X_1\}$ has the highest scale, we have $a_n$ scales no higher than $\O(1)$, that is, there exists constants $c_n$'s such that $a_n \le c_n$. Thus $\sum_{i\neq j}{\rm Cov}\{X_i,X_j\}=\O(1/M^{\alpha})$, and consequently ${\rm Var}\{Y\}$ scales no higher than $1/M^\alpha$.
\end{proof}

Given Lemma \ref{lemma:var}, we only need to find the condition for the variances of $(K-1)P_{i,e}/(P_c^4M^{1-r_s})$, $P_{n,e}/(PP_c^4M^{1-r_s})$, $P_{e,1}/(P_c^4M^{1-r_s})$, $P_{e,2}/(P_c^4M^{1-r_s})$, and $\frac{1}{P_c^4M^{1-r_s}}P_{e,3}$ to scale no larger than $1/M$. Using the results on the variances of SINR components, the variances of the terms can be obtained as
\begin{eqnarray*}
&&{\rm Var}\{\frac{K-1}{P_c^4M^{1-r_s}}P_{i,e}\}=\frac{(K-1)^2}{P_c^2M^{2-2r_s}}\left(\hspace{-1mm}\frac{4}{K-1}\hspace{-1mm}+\hspace{-1mm}
\frac{8+10P_c}{P_cM}\hspace{-1mm}+\hspace{-1mm}\frac{K^2+18(K-2)P_c}{(K-1)P_c^2M^2}\hspace{-1mm}\right)\sim \O\left(M^{-(2-2r_s-2r_c-r_k)}\right),
\\
&&{\rm Var}\{\frac{1}{PP_c^4M^{1-r_s}}P_{n,e}\}=\frac{\frac{2}{P_c^3}+\frac{5}{P_c^2}-\frac{2}{P_c}}{M^{3-2r_s}P^2}\sim \O\left(M^{-(3-2r_s-3r_c-2r_p)}\right), \\
&&{\rm Var}\{\frac{1}{P_c^4M^{1-r_s}}P_{e,1}\}=\frac{3K}{M^{4-2r_s}}(\frac{1}{P_c}-1)^4
\sim \O\left(M^{-(4-2r_s-4r_c-r_k)}\right), \\
&&{\rm Var}\{\frac{1}{P_c^4M^{1-r_s}}P_{e,2}\}={\rm Var}\{\frac{1}{P_c^4M^{1-r_s}}P_{e,3}\}=\frac{1}{M^{2-2r_s}}(\frac{1}{P_c}-1)^2
\sim \O\left(M^{-(2-2r_s-2r_c)}\right),
\end{eqnarray*}
where the scaling behaviour at the end of each line is obtained from the definitions of the scaling exponents in (\ref{exponents-def}) and considering the constraints in (\ref{cond-scaling}). Then, we can see that the condition for the scaling order of each term to be no higher than $1/M$ is that both following constrains are satisfied.
\begin{equation}
r_k+2r_c+2r_s\le 1,\quad 2r_p+3r_c+2r_s \le 2. \label{eq:cond_2}
\end{equation}
Combining (\ref{cond-scaling}) and (\ref{eq:cond_2}), we get the sufficient condition for the SINR to be deterministic in (\ref{eq:suff-con}).

\subsection{Proof of Proposition \ref{pro:linearSINR}}
\label{app-C}
Linearly increasing SINR means that the SINR scaling exponent is 1, i.e., $r_s=1$. Thus the SINR can be formulated as 
\begin{equation*}
{\rm SINR}_{i,e}= M\frac{P_{s,e}/P_c^4 }{P_{i,e}\frac{K-1}{P_c^4}+\frac{1}{PP_c^4}P_{n,e}+\frac{1}{P_c^4}(
P_{e,1}+P_{e,2}+P_{e,3})+\frac{K(1+\frac{K}{MP_c}+\frac{1}{PP_cM})}{QP_c}}.
\end{equation*}
From the SINR scaling law in (\ref{SNR-scaling}), we can see that the sufficient and necessary condition for $r_s=1$ is $r_c=r_p=r_k=r_q=0$ (note that $r_c,r_p,r_q,r_k \in [0,1]$). 

With the parameter values, we can calculate that the SCVs of $P_{s,e}/P_c^4$ and $P_{n,e}/P/P_c^4$ scales of $1/M$. Therefore, they are asymptotically deterministic and can be approximated with their mean values. On the other hand, the SCVs of $(K-1)P_{i,e}/P_c^4$, $P_{e,1}/P_c^4$, $P_{e,2}/P_c^4$, and $P_{e,3}/P_c^4$ are constant. We analyze their behaviour next.  
$${\rm Var}\left\{\frac{K-1}{P_c^4}P_{i,e}\right\}\ge \frac{4(K-1)}{P_c^2}\ge 4(K-1) \left(\frac{1}{P_c}-1\right)^2=4(K-1){\rm Var}\left\{\frac{P_{e,2}}{P_c^4}\right\}.$$
 Also, $P_{e,2}$ and $P_{e,3}$ have the same distribution. As we mainly consider the non-trivial case that $K\ge 3$, we have ${\rm Var}\{(K-1)P_{i,e}/P_c^4\}\gg P_{e,2}/P_c^4, P_{e,3}/P_c^4$, especially when the CSI quality $P_c$ is high.
  Besides, the mean of $P_{e,1}/P_c^4$ scales as $1/M$, and its variance scales as $1/M^2$. Thus the variance of this term is also much smaller than $P_{i,e}(K-1)/P_c^4$. Therefore, $P_{i,e}(K-1)/P_c^4$ dominates the random behaviour of the SINR and other terms can  be approximated with their mean values. Thus the SINR approximation in (\ref{eq:SINR_app_high}) is obtained, where only dominant terms of $M$ are kept.

\subsection{Proof of Proposition \ref{pro:pdf}}
\label{app-D}
When $K=2$, $P_{i,e} =\left|\frac{{\bf g}_i\hat{\bf g}_i^H\hat{\bf f}_i^H{\bf f}_k}{\sqrt{M^3}}+\frac{{\bf g}_i\hat{\bf g}_k^H\hat{\bf f}_k^H{\bf f}_k}{\sqrt{M^3}}\right|^2/(K-1).$
Then, using CLT, $P_{i,e}$ has an exponential distribution with parameter $1/d_e$. Then, the PDF can be approximated as $f_{P_{i,e}}(y)\approx e^{-y/d_e}/d_e$, which is the same as (\ref{cf-pdf-e}) for $K=2$.

Now, we work on the more complicated case of $K\ge 3$. Firstly,
\begin{eqnarray*}
&&\hspace{-5mm}\frac{|{\bf g}_i\hat{\bf G}^{\hspace{-0.5mm}H}\hspace{-0.5mm}\hat{\bf F}^{\hspace{-0.5mm}H}{\bf f}_k|^2}{M^3}\hspace{-1mm} =\left|\frac{{\bf g}_i\hat{\bf g}_i^H\hat{\bf f}_i^H{\bf f}_k}{\sqrt{M^3}}\hspace{-1mm}+\hspace{-1mm}\frac{{\bf g}_i\hat{\bf g}_k^H\hat{\bf f}_k^H{\bf f}_k}{\sqrt{M^3}}\hspace{-1mm}+\hspace{-2mm}\sum_{n\neq i,n\neq k}^M\hspace{-2mm}\frac{{\bf g}_i\hat{\bf g}_n^H\hat{\bf f}_n^H{\bf f}_k}{\sqrt{M^3}}\right|^2
\end{eqnarray*}
With the help of CLT, as $M\gg1$, $\frac{{\bf g}_i\hat{\bf g}_i^H\hat{\bf f}_i^H{\bf f}_k}{\sqrt{M^3}}$ is approximately distributed as $CN(0,P_c^3+\frac{P_c^2}{M})$, and $\frac{{\bf g}_i\hat{\bf g}_n^H\hat{\bf f}_n^H{\bf f}_k}{\sqrt{M^3}}$ is approximately distributed as $CN(0,\frac{P_c^2}{M})$. We can further show that the covariances between $\frac{{\bf g}_i\hat{\bf g}_i^H\hat{\bf f}_i^H{\bf f}_k}{\sqrt{M^3}}$, $\frac{{\bf g}_i\hat{\bf g}_n^H\hat{\bf f}_n^H{\bf f}_k}{\sqrt{M^3}}$, and $\frac{{\bf g}_i\hat{\bf g}_k^H\hat{\bf f}_k^H{\bf f}_k}{\sqrt{M^3}}$ are zero, thus they are uncorrelated. For tractable analysis, we assume independence as they are Gaussian distributed. Now we conclude that $\frac{|{\bf g}_i\hat{\bf G}^{H}\hat{\bf F}^{H}{\bf f}_k|^2}{(K-1)M^3}$ has a gamma distribution with shape parameter $1$ and scale parameter $\frac{P_c^3}{K-1}\left(2+\frac{K}{MP_c}\right)$, which is also defined as $d_e$.

Using CLT, the covariance between $\frac{|{\bf g}_i\hat{\bf G}^{H}\hat{\bf F}^{H}{\bf f}_k|^2}{(K-1)M^3}$ and $\frac{|{\bf g}_i\hat{\bf G}^{H}\hat{\bf F}^{H}{\bf f}_l|^2}{(K-1)M^3}$ ($k\neq l$) can be derived as
\begin{equation}
{\rm Cov}=\frac{4P_c^5+10P_c^6}{(K-1)^2M}+\frac{18P_c^5+(2K-4)P_c^6}{(K-1)^2M^2},\label{eq:Inter_cov-e}
\end{equation}
where the proof is omitted due to save space. The correlation coefficient between the two is subsequently
\begin{eqnarray}
\rho_{jl}&=&\hspace{-1mm}\frac{{\rm Cov}\left\{\frac{|{\bf g}_i\hat{\bf G}^{H}\hat{\bf F}^{H}{\bf f}_k|^2}{(K-1)M^3},\frac{|{\bf g}_i\hat{\bf G}^{H}\hat{\bf F}^{H}{\bf f}_l|^2}{(K-1)M^3}\hspace{-1mm}\right\}}{\sqrt{\mathrm{Var}
\Big\{\frac{|{\bf g}_i\hat{\bf G}^{H}\hat{\bf F}^{H}{\bf f}_k|^2}{(K-1)M^3}\Big\}\mathrm{Var}\Big\{\frac{|{\bf g}_i\hat{\bf G}^{H}\hat{\bf F}^{H}{\bf f}_l|^2}{(K-1)M^3}\Big\}}}\hspace{-1mm}\approx \frac{1}{M}\frac{\frac{4}{P_c}+10}{(2+\frac{K}{MP_c})^2}.\label{eq:correlation-e}
\end{eqnarray}
It equals $\rho_e^2$ based on the definition in (\ref{eq:miu}).

Thus $P_{i,e}$ is a summation of $K-1$ correlated random variables following the same Gamma distribution. From Corollary 1 of \cite{refe2}, the PDF of $P_{i,e}$ is 
\begin{equation}\label{pdfproof}
f_{P_{i,e}}(y)=\prod_{i=1}^{K-1}\Big(\frac{\sigma_1}{\sigma_i}\Big)\sum_{j=0}^\infty
\frac{\delta_jy^{K+j-2}e^{-y/\sigma_1}}{\sigma_1^{K+j-1}\Gamma(K+j-1)},
\end{equation}
where $\sigma_1\le\sigma_2\le\cdots\le \sigma_{K-1} $ are the ordered eigenvalues of the $(K-1)\times (K-1)$ matrix $\bf A$, whose diagonal entries are $d_e$ and off-diagonal entries are $d_e\rho_e$, and $\delta_j$'s are defined iteratively as
\begin{eqnarray}
&&\delta_0  \triangleq1, \quad \delta_{j+1}  \triangleq \frac{1}{j+1}\sum\limits_{m=1}^{j+1}\left[\sum\limits_{n=1}^{K-1}
\Bigg(1-\frac{\sigma_1}{\sigma_n}\Bigg)^m\right]\delta_{j+1-m}.\label{deltai+1}
\end{eqnarray}
As $\bf A$ is a circulant matrix whose off-diagonal entries are the same, its eigenvalues can be calculated as
\begin{eqnarray} 
&&\sigma_1=\cdots
=\sigma_{K-2}=d_e-d_e\rho_e,\label{lamda1} \hspace{1cm}\sigma_{K-1}\hspace{-1mm}=d_e+(K-2)d_e\rho_e. \label{lamdaK}
\end{eqnarray}
Then we can show that
\begin{equation}\label{deltai}
\delta_j = \Bigg(\frac{(K-1)\rho_e}{1+(K-2)\rho_e}\Bigg)^j=\left(\frac{b_e}{b_e+c_e}\right)^j.
\end{equation}
Substituting \eqref{lamda1} and \eqref{deltai} into \eqref{pdfproof}, we can get PDF of 
$P_{i,e}$ as in (\ref{equ4}) in Proposition \ref{pro:pdf}. Notice that
\setlength{\arraycolsep}{1pt}
\begin{eqnarray*}
&&\hspace{-2mm}\sum_{i=0}^{\infty}\left(\frac{b_e}{b_e+c_e}\right)^i\phi\hspace{-0.5mm}\left(y;K\hspace{-0.5mm}+\hspace{-0.5mm}i\hspace{-0.5mm}-\hspace{-0.5mm}1,d_ec_e\right)=\left(\frac{b_e}{b_e+c_e}\right)^{\hspace{-0.5mm}-\hspace{-0.5mm}(K\hspace{-0.5mm}-\hspace{-0.5mm}2)}\hspace{-2mm}\frac{e^{-\frac{y}{d_ec_e}}}{d_ec_e} \hspace{-1mm}\left(\hspace{-1mm}\sum_{n=0}^\infty \hspace{-2mm}-\hspace{-2mm} \sum_{n=0}^{K-3}\right)\hspace{-1mm} \left(\frac{b_e}{d_ec_e(b_e+c_e)}\right)^{\hspace{-1mm} n}\hspace{-2mm}\frac{y^n}{n!}.
\end{eqnarray*}
By straightforward calculations, we can obtain 
the closed-form PDF of $P_{i,e}$ in (\ref{cf-pdf-e}).

\linespread{1.2}

\end{document}